\documentclass[11pt]{article}

\usepackage{setspace}
\usepackage{cprotect}
\usepackage{amsmath,amssymb,amsthm}
\usepackage[noend]{algorithmic}
\usepackage[ruled,vlined]{algorithm2e}
\usepackage{hyperref}
\usepackage{fullpage}
\usepackage{bbm}
\usepackage{cite}
\usepackage{makeidx}
\usepackage{enumerate}
\usepackage{graphicx,float,psfrag,epsfig}
\usepackage{epstopdf}
\usepackage{color}
\usepackage{enumitem}
\usepackage{subfig}
\usepackage{caption}
\usepackage{bigints}
\usepackage{mathtools}
\usepackage[mathscr]{euscript}
\usepackage{tikz}
\usetikzlibrary{shapes.geometric}
\usetikzlibrary{arrows.meta,arrows}

\makeatletter
\def\blfootnote{\gdef\@thefnmark{}\@footnotetext}
\makeatother

\title{Bayes-Optimal Estimation in Generalized Linear Models \\  
via Spatial Coupling}
\author{Pablo {Pascual Cobo}, Kuan Hsieh, Ramji Venkataramanan}

\def\<{\langle}
\def\>{\rangle}

\definecolor{maroon}{rgb}{0.5, 0.0, 0.0}

\newtheorem{theorem}{Theorem}
\newtheorem{lemma}{Lemma}[section]
\newtheorem{defi}{Definition}[section]

\newtheorem{remark}{Remark}[section]
\newtheorem{prop}{Proposition}[section]
\newtheorem{corr}{Corollary}[section]

\newcommand{\abs}[1]{\left \lvert #1\right \rvert}
\DeclareMathOperator*{\argmin}{arg\,min}
\DeclareMathOperator*{\argmax}{arg\,max}
\usepackage{mathtools}

\newcommand{\bA}{\boldsymbol{A}}

\newcommand{\bW}{\boldsymbol{W}}

\newcommand{\bM}{\boldsymbol{M}}

\newcommand{\bh}{\boldsymbol{h}}
\newcommand{\be}{\boldsymbol{e}}
\newcommand{\bx}{\boldsymbol{x}}
\newcommand{\bz}{\boldsymbol{z}}
\newcommand{\by}{\boldsymbol{y}}

\newcommand{\bq}{\boldsymbol{q}}
\newcommand{\bp}{\boldsymbol{p}}
\newcommand{\bi}{\boldsymbol{i}}

\newcommand{\tf}{\tilde{f}}
\newcommand{\tg}{\tilde{g}}
\newcommand{\tq}{\tau^q}
\newcommand{\tp}{\tau^p}
\newcommand{\htq}{\hat{\tau}^q}
\newcommand{\htp}{\hat{\tau}^p}

\newcommand{\veps}{\varepsilon}
\newcommand{\ups}{\upsilon}
\newcommand{\bveps}{\boldsymbol{\varepsilon}}
\newcommand{\bbeta}{\boldsymbol{\beta}}

\newcommand{\bgamma}{\boldsymbol{\gamma}}

\newcommand{\sB}{\textsf{B}}
\newcommand{\sD}{\textsf{D}}
\newcommand{\sT}{\textsf{T}}

\newcommand{\C}{\textsf{C}}
\newcommand{\R}{\textsf{R}}
\newcommand{\Lc}{\textsf{C}}
\newcommand{\Lr}{\textsf{R}}
\newcommand{\sfr}{\textsf{r}}
\newcommand{\sfc}{\textsf{c}}

\newcommand{\delin}{\normalfont \delta_\text{in}}

\newcommand{\bttq}{\boldsymbol{\alpha^q}}
\newcommand{\bttp}{\boldsymbol{\alpha^p}}

\newcommand{\om}{\omega}
\newcommand{\lam}{\Lambda}
\newcommand{\Aiid}{\bA_{\rm{iid}}}

\newcommand{\gin}{{g_{\normalfont \textrm{in}}}}
\newcommand{\ginpr}{{g'_{\normalfont \textrm{in}}}}
\newcommand{\gout}{{g_{\normalfont \textrm{out}}}}
\newcommand{\goutpr}{{g'_{\normalfont \textrm{out}}}}
\newcommand{\ginbayes}{{g^*_{\normalfont \textrm{in}}}}
\newcommand{\goutbayes}{{g^*_{\normalfont \textrm{out}}}}
\newcommand{\bargin}{{\bar{g}_{\normalfont \textrm{in}}}}
\newcommand{\bargout}{{\bar{g}_{\normalfont \textrm{out}}}}

\newcommand{\psiin}{\psi_{\normalfont \textrm{in}}}
\newcommand{\psiout}{\psi_{\normalfont \textrm{out}}}

\newcommand{\pot}{U}
\newcommand{\frs}{f_{\normalfont \textrm{RS}}}

\newcommand{\mc}{\mathcal}
\newcommand{\normal}{\mathcal{N}}
\newcommand{\reals}{\mathbb{R}}

\newcommand{\beq}{\begin{equation}}
\newcommand{\eeq}{\end{equation}}

\newcommand{\sign}{\textrm{\sign}}

\newcommand{\E}{\mathbb{E}}
\newcommand{\var}{\normalfont \textrm{Var}}
\newcommand{\mmse}{{\sf{mmse}}}

\newcommand{\MMSE}{\textsf{MMSE}}
\newcommand{\disteq}{\stackrel{\rm d}{=}}
\newcommand{\Pout}{P_{\normalfont \textrm{out}}}

\begin{document}

\date{}
\maketitle

\blfootnote{This work was supported in part by an Engineering and Physical Sciences Research Council Doctoral Training Award. This paper was presented in part at the 2023 IEEE International Symposium on Information Theory.
P.~Pascual~Cobo and R.~Venkataramanan are with the Department of Engineering, University of Cambridge,  UK ({\tt pp423@cam.ac.uk}; {\tt rv285@cam.ac.uk}).
K.~Hsieh was with the Department of Engineering, University of Cambridge,  UK 
({\tt kuanhsieh95@gmail.com}).}

\begin{abstract}
  We consider the problem of signal estimation in a generalized linear model (GLM). GLMs include many canonical problems in statistical estimation, such as  linear regression, phase retrieval, and 1-bit compressed sensing. Recent work has precisely characterized the asymptotic minimum mean-squared error (MMSE)  for  GLMs  with i.i.d.~Gaussian sensing matrices. However, in many models there is a significant gap between the MMSE and the performance of the best known feasible estimators.  We address this issue by considering GLMs defined via \emph{spatially coupled} sensing matrices. We propose an efficient approximate message passing (AMP) algorithm for estimation and prove that with a simple choice of spatially coupled design, the MSE of a carefully tuned AMP estimator approaches the asymptotic MMSE as the dimensions of the signal and the observation grow proportionally.
To prove the result, we first rigorously characterize the asymptotic  performance of AMP for a GLM with a generic spatially coupled design. This characterization is in terms of a deterministic recursion (`state evolution') that depends on the parameters defining the spatial coupling. Then, using a simple spatially coupled design and a judicious choice of  functions for the AMP algorithm, we analyze the fixed points of the resulting state evolution and show that it achieves the asymptotic MMSE. 
Numerical results for phase retrieval and rectified linear regression  show that spatially coupled designs can yield substantially lower MSE than i.i.d.~Gaussian designs at finite  dimensions when used with AMP algorithms.  
  \end{abstract}

  \section{Introduction}
\label{sec:intro}
Consider a generalized linear model (GLM), where the goal is to estimate a signal $\bx \in \reals^n$ from an observation  $\by \in \reals^m$  obtained as: 
\begin{equation}
    \by = \varphi( \bz, \, \bveps), \quad \text{ where } \bz = \bA \bx.
    \label{eq:GLM_def}
\end{equation}
Here  $\bA \in  \reals^{m \times n}$ is a known sensing matrix, $\bveps \in \reals^m$ is an unknown  noise vector and $\varphi: \reals^2 \to \reals$ is a known function applied row-wise to the input. The model \eqref{eq:GLM_def} covers many widely studied problems in statistical estimation and signal processing, including linear regression \cite{Donoho1,eldar2012compressed} ($\by = \bA \bx  + \bveps$), phase retrieval \cite{shechtman2015phase,fannjiang2020numerics} ($\by = \abs{\bA \bx}^2  + \bveps$), and 1-bit compressed sensing \cite{boufounos20081} ($ \by = \text{sign}( \bA \bx + \bveps)$).

In this paper, we consider the high-dimensional setting where $m,n$ are both large with $\lim_{n \to \infty} \frac{m}{n}$ $=\delta$, a positive constant that does not scale with $n$.
The components of the signal $\bx$ are assumed to be distributed according to a prior $P_{X}$, and the components of the noise $\bveps$   according to a distribution $P_{\bar{\veps}}$.  
If the entries of $\bA$ are also assumed to be drawn from a given distribution,  the  minimum mean-squared error is:
\begin{equation}
   \MMSE_n :=  \frac{1}{n} \E\{ \| \bx - \E\{\bx \mid \bA, \, \by \} \|^2\}.
   \label{eq:MMSE_def}
\end{equation}
Two key questions in this setting are: i)  What is the limiting behavior of  $\MMSE_n$ as $m,n \to \infty$ with their ratio fixed at a given $\delta$?
ii) Noting that the MMSE estimator $\E\{\bx \, | \, \bA, \, \by \} $ is computationally infeasible for large $n$, what is the smallest mean-squared error (MSE) achievable by  efficient estimators? 

The first question has been precisely answered for i.i.d.~Gaussian sensing matrices, first for linear models \cite{reeves2019thereplica,barbier2020mutual} and then for generalized linear models in \cite{barbier2019optimal}. In \cite{barbier2019optimal}, Barbier et al. proved that under mild technical conditions, $\MMSE_n$ converges to a well-defined limit  that can be numerically computed (see Appendix \ref{app:pot_fun_equiv}). Regarding computationally efficient estimators, a
variety of estimators based on convex relaxations, spectral methods, and non-convex methods have been proposed for specific GLMs, such as sparse linear regression \cite{Tibs96,hastie2019statistical}, phase retrieval \cite{netrapalli2013phase,mondelli2017fundamental,lu2020phase} and one-bit compressed sensing \cite{plan2013one, jacques2013robust}. Most of these techniques are generic and  are not well-equipped to exploit the signal prior, so their estimation error (MSE) is generally much higher than the MMSE.

 Approximate message passing (AMP) is a family of iterative algorithms which can be tailored to take advantage of the signal prior. AMP algorithms were first proposed for estimation in linear models \cite{Kab03,donoho2009message,bayati2011thedynamics,krzakala2012statistical} and have since been applied to a range of statistical estimation problems, including GLMs and their variants \cite{rangan2010generalized,schniter2014compressive,ma2019optimization, maillard2020phase, Tan23c}, and low-rank matrix and tensor estimation \cite{deshpande2014information,fletcher2018iterative,lesieur2017constrained,MV21estimation, LiFanWei_Z2, rossetti2023approximate}. 
 An attractive feature of AMP is that under suitable model assumptions, its high-dimensional estimation performance  is precisely characterized by a succinct deterministic recursion called \emph{state evolution}.   Using the state evolution analysis, it has been proved that  the AMP is Bayes-optimal for certain problems such as symmetric rank-1 matrix estimation \cite{deshpande2014information,MV21estimation}, i.e., its MSE converges to the limiting MMSE.  
 However, for many  GLMs, including phase retrieval,  the MSE of AMP can be substantially higher than the limiting MMSE  \cite{barbier2019optimal}. 

 In this work, we show that for any GLM satisfying certain mild conditions, the limiting MMSE corresponding to an i.i.d. Gaussian sensing matrix can be achieved efficiently using a \emph{spatially coupled} sensing matrix and an AMP algorithm for estimation. A spatially coupled Gaussian sensing matrix is composed of blocks with different variances; it consists of  zero-mean Gaussian entries that are independent  but not identically distributed across different blocks. 
 To ensure a fair comparison with i.i.d.~Gaussian matrices and their limiting MMSE, the variances in 
the spatially coupled sensing matrix $\bA$ are chosen so that the signal strength $ \frac{1}{m}\E\{ \| \bA \bx \|^2 \}$ is the same as for the i.i.d.~Gaussian design.

Spatial coupling was introduced by Felstrom and Zigangirov in the context of LDPC codes  in  \cite{felstrom1999time}. Spatially coupled LDPC codes have since been shown to achieve the capacity of a large class of binary-input channels \cite{kudekar2011threshold, kudekar2013spatially}. An overview of spatially coupled LDPC codes can be found in \cite{costello2014spatially}, and   we refer the reader to \cite{mitchell2015spatially} for a list of references. Spatially coupled matrices for compressed sensing were introduced by Kudekar and Pfister \cite{kudekar2010theeffect}, and then studied in a number of works \cite{krzakala2012statistical,takeuchi2011cdmaspatialcoupling, takeuchi2015performance, donoho2013information}. 
Donoho et al.  \cite{donoho2013information} proved that the Bayes-optimal  error of linear models with standard Gaussian design can be achieved using spatially coupled sensing matrices  with estimation via AMP. Specifically, they showed that for noiseless linear models, perfect signal recovery is possible with spatial coupling and AMP provided the sampling ratio $\delta$ exceeds the R{\'e}nyi information dimension of the  signal prior. Spatially coupled matrices have also been used for constructing rate-optimal communication schemes for both AWGN channels  \cite{barbier2015approximate, barbier2017approximate,venkataramanan19monograph,rush2020capacity,hsieh2022GMAC} and  general memoryless  channels \cite{barbier2019universal}.

The focus of the above works is mainly on spatial coupling for estimation in \emph{linear} models. Our work significantly expands the scope of spatial coupling beyond linear models, and shows how to efficiently achieve the Bayes-optimal error  for  GLMs, including canonical examples such as phase retrieval. We must mention that to use spatial coupling we need some flexibility in constructing the sensing matrix, which is a reasonable assumption in many signal processing and imaging applications.  


\paragraph{Structure of the paper and main contributions.} 
In Section \ref{sec:spatially-coupled-glms}, we describe the construction of a spatially coupled Gaussian sensing matrix and present the AMP algorithm for estimation in a   GLM defined via a spatially coupled matrix. Our first result,  Theorem \ref{thm:gen_sc_gamp_SE}, gives performance guarantees for AMP applied to a GLM with a generic spatially coupled design. It shows that in each AMP iteration, the joint empirical distribution of the signal and AMP estimate converges to the law of a pair of random variables, defined in terms of the signal prior, the noise distribution and the deterministic state evolution parameters. The state evolution  parameters are computed via a  recursion that depends on the `denoising' functions used to define the AMP algorithm and on the parameters used to define the spatially coupled design. 

In Section \ref{sec:mmse-fp-bayes}, we state our second main result, Theorem \ref{thm:SC_FPs}, which 
shows that for any GLM and signal prior (that satisfy certain mild regularity conditions), a simple spatially coupled design with AMP estimation can efficiently achieve an MSE arbitrarily close to the limiting  Bayes-optimal error for an i.i.d. design (with the same signal-to-noise ratio). This is done with specific choices for denoising functions, and suitably large choice of coupling parameters. Theorem \ref{thm:SC_FPs} also establishes the limiting MSE achieved by AMP with an i.i.d.~Gaussian design, thereby characterizing the gap between the performance of spatially coupled and i.i.d.~designs.

In Section \ref{sec:numerical_results}, we present numerical results for phase retrieval and rectified linear regression, which confirm  that the MSE can be  significantly smaller with a spatially coupled design compared to an i.i.d.~Gaussian one, with AMP being used for estimation in both settings. Section \ref{sec:proofs} contains the proofs of the main results, and Section \ref{sec:conclusion} concludes the paper. 

\paragraph{Technical Ideas.} The AMP performance characterization in Theorem \ref{thm:gen_sc_gamp_SE}, for a GLM with a generic spatially coupled design, is proved using a change of variables that maps the proposed  algorithm to an abstract AMP iteration with matrix-valued iterates. A state evolution characterization for the abstract AMP iteration is obtained using the results in  \cite{donoho2013information}, and then translated via the change of variables to obtain the state evolution characterization for the proposed AMP. 

For our second result, Theorem \ref{thm:SC_FPs}, we use  specific choices for the spatially coupled design and the AMP denoising functions. The spatially coupled matrix has a band-diagonal structure, with all the non-zero entries being independent and identically distributed. The denoising functions are chosen based on the Bayes-optimal choices in the uncoupled setting, and adapted to the spatially coupled design. Theorem \ref{thm:SC_FPs} is obtained by analyzing the fixed point of the state evolution recursion obtained with the above choices. A key tool in this analysis is a result of Yedla et al. \cite{yedla2014asimple} characterizing the fixed points of a general coupled recursion in terms of a suitable `potential' function. 

We note that both our spatially coupled design and the proof of Bayes optimality of AMP (Theorem \ref{thm:SC_FPs}) are substantially simpler than those in \cite{donoho2013information} for spatially coupled linear models. In \cite{donoho2013information}, the fixed points of the  coupled state evolution  recursion for the linear model were analyzed using a continuum version of the state evolution and a perturbation argument. Our simpler proof  is facilitated by the powerful result of \cite{yedla2014asimple} which gives a recipe for constructing potential functions for a general class of coupled recursions.

\paragraph{Notation.} We write $[n]$ for the set $\{1, \ldots, n \}$.
For a vector $\bx$, we write $x_i$ for the $i$th component. If $\bx$ is a matrix, $x_i$ denotes its $i$th row. We write $\mathbf{1}_n$ for the all-ones vector of length $n$, and $\mathbf{0}_n$ for the all-zeros vector.  For random variables $X, Y$, we write $X \perp Y$ to denote that they are independent.

\section{Spatially Coupled GLMs}
\label{sec:spatially-coupled-glms}

\subsection{Model assumptions} \label{subsec:model_assump}
As $n, m \to\infty$, we assume that $ \frac{m}{n}  \to \delta$, for some constant $\delta > 0$. Both the signal $\bx$ and the noise vector $\bveps$ are independent of the sensing matrix $\bA$.  As $n \to \infty$, the empirical distributions of the signal and the noise vectors are assumed to converge in Wasserstein distance to well-defined limits. More precisely, let $\nu_n(\bx)$ and $\nu_m(\bveps)$ denote the empirical distributions of $\bx$ and $\bveps$, respectively. Then for some $k \in [2, \infty)$,  there exist scalar random variables $X \sim P_{X}$  and $\bar{\veps}\sim P_{\bar{\veps}}$ with $\E\{ |X|^k \},\E\{ |\bar{\veps}|^k \}<\infty$, such that  $d_k(\nu_n(\bx),P_{X}) \rightarrow 0$ and $d_k(\nu_m(\bveps), P_{\bar{\veps}}) \rightarrow 0$ almost surely, where $d_k(P,Q)$ is the $k$-Wasserstein distance between distributions $P, Q$ defined on the same Euclidean probability space.  We note that this assumption on the empirical distributions of $\bx$ and $\bveps$ is more general than assuming that their entries are i.i.d.~according to $P_X$ and 
$P_{\bar{\veps}}$, respectively (which it includes as a special case). 

To avoid sign-symmetry issues, we assume that \emph{at least} one of the following two conditions holds: 
\begin{align}
  \mathbb{P}\left(\varphi(Z, \bar{\veps}) = \varphi(-Z, \bar{\veps}) \right) \, \neq  \, 1, \ \  \text{ where } \  Z \sim \normal(0,1) \perp \bar{\veps} \sim P_{\bar{\veps}} \, ,   
  \label{eq:sign_anti_sym}
\end{align}
or
\begin{align}
    \E\{ X\} \neq 0.
    \label{eq:exp_ne_0}
\end{align}
When the condition \eqref{eq:sign_anti_sym} does not hold, as in phase retrieval  where $Y=\abs{Z}^2+ \bar{\veps}$, AMP requires an initialization with non-zero asymptotic empirical correlation with the signal $\bx$. In this case, \eqref{eq:exp_ne_0} ensures that the natural initialization $\E\{X\} \mathbf{1}_n$ has non-zero asymptotic empirical correlation with the signal. 
For GLMs which satisfy the condition in \eqref{eq:sign_anti_sym}, e.g.,  linear models ($Y= Z + \bar{\veps}$) and rectified linear regression ($Y= \max\{0, Z+ \bar{\veps}\}$), a random initialization suffices for AMP to produce non-trivial estimates of the signal. For GLMs with i.i.d.~Gaussian designs, when neither of the two conditions  holds, AMP can be initialized with a spectral estimator \cite{mondelli2021approximate}.  Spectral estimators for spatially coupled designs are beyond the scope of this paper, and we leave this as a direction for future work.

\subsection{Spatially Coupled Sensing Matrix}
\label{subsec:sc-design-matrix}
Consider the GLM \eqref{eq:GLM_def} with  sensing matrix $\bA \in\mathbb{R}^{m \times n}$.
A spatially coupled Gaussian sensing matrix 
$\bA$ consists of independent zero-mean normally distributed entries whose variances are specified by a \emph{base matrix} $\bW$ of dimension $\Lr \times \Lc$.
The sensing matrix $\bA$ is obtained by replacing each entry of the base matrix $W_{\sfr \sfc}$ 
by an $\frac{m}{\Lr} \times \frac{n}{\Lc}$ matrix with entries drawn i.i.d.~$\sim \mc{N}(0, \frac{W_{\sfr \sfc}}{m/\Lr})$, 
for $\sfr\in[\Lr]$, $\sfc\in[\Lc]$. This is similar to the ``graph lifting'' procedure for constructing spatially coupled LDPC codes from protographs \cite{mitchell2015spatially}.
See Figure \ref{fig:spatial_coupling_example} for an example.

From the construction, $\bA$ has independent Gaussian entries
\begin{equation}
    \label{eq:construct_sc_A}
A_{ij} \sim \mc{N}\bigg(0,\frac{1}{m/\Lr} 
W_{\sfr(i), \sfc(j)} \bigg),  \quad \text{for }  i \in [m], \ j\in [n].
\end{equation}
Here the operators $\sfr(\cdot):[m]\rightarrow[\Lr]$ and $\sfc(\cdot):[n]\rightarrow[\Lc]$ in \eqref{eq:construct_sc_A} map a particular row or column index to its corresponding \emph{row block} or \emph{column block} index in $\bW$. Each entry of the base matrix corresponds to an $\frac{m}{\Lr} \times \frac{n}{\Lc}$ block of the sensing matrix $\bA$, and each block can be viewed as an (uncoupled) i.i.d.~Gaussian sensing matrix with sampling ratio
\begin{equation}
\label{eq:del_inner}
\delin := \frac{m/\Lr}{n/\Lc}
= \frac{\Lc}{\Lr} \, \delta.
\end{equation}
\begin{figure}[!t]
\centerline{\begin{tikzpicture}[scale=0.2, font=\large]
       \foreach \y in {0,...,9}
               \draw[thick, color=black!25] (10,-2*\y) -- (31,-2*\y);
        \foreach \x in {0,1, 2, 3, 4, 5, 6}
                \draw[thick, color=black!25] ({10+3*\x},0) -- ({10+3*\x},-18); 
        \foreach \x in {0,1, 2, 3, 4, 5, 6}{
                \draw[fill=cyan!65, thick] (10+3*\x, {-2*\x}) rectangle (13+3*\x, {-2-2*\x});
                \draw[fill=cyan!65, thick] (10+3*\x, {-2-2*\x}) rectangle (13+3*\x, {-4-2*\x});
                \draw[fill=cyan!65, thick] (10+3*\x, {-4-2*\x}) rectangle (13+3*\x, {-6-2*\x});
        }
        \draw[thick] (10, 0) -- (31, 0) -- (31, -18) -- (10, -18) -- cycle;
        \draw (20.5, -18.5) node[anchor=north] {\small{Sensing matrix $\bA$}};
        \draw[to-to] (9.5, 0) -- (9.5, -2);
        \draw (9.5, -1) node[anchor=east] {$\frac{m}{\R}$};
        \draw[to-to] (6.8, 0) -- (6.8, -18);
        \draw (5.8, -5.5) node[anchor=east, rotate=90] {$m$};
        \draw[to-to] (10, 0.5) -- (13, 0.5);
        \draw (11.5, 0.3) node[anchor=south] {\small{$n/\C$}};
        \draw[to-to] (10, 2.8) -- (31, 2.8);
        \draw (20.5, 2.5) node[anchor=south] {$n$};
       \foreach \y in {0,...,9}
               \draw[thick, color=black!25] (36,-4-\y) -- (43,-4-\y);
        \foreach \x in {0,1, 2, 3, 4, 5, 6}
                \draw[thick, color=black!25] ({36+\x},-4) -- ({36+\x},-13);       
        \draw[thick] (36, -4) -- (43, -4) -- (43, -13) -- (36, -13) -- cycle;
        \draw (39.5, -14) node[anchor=north] {\small{Base matrix $\bW$}};
        \draw[to-to] (43.9, -4) -- (43.9, -13);
        \draw (43.9, -8.5) node[anchor=west] {$\R$};
        \draw[to-to] (36, -3.3) -- (43, -3.3);
        \draw (39.5, -3.3) node[anchor=south] {$\C$};
        \foreach \x in {0,1, 2, 3, 4, 5, 6}{
                \draw[fill=cyan!65, thick] (36+\x, {-4-\x}) rectangle (37+\x, {-5-\x});
                \draw[fill=cyan!65, thick] (36+\x, {-5-\x}) rectangle (37+\x, {-6-\x});
                \draw[fill=cyan!65, thick] (36+\x, {-6-\x}) rectangle (37+\x, {-7-\x});
        }
        \draw[fill=maroon!70, thick] (37,-5) rectangle (38,-6);
        \draw[fill=maroon!70, thick] (13,-2) rectangle (16,-4);
        \draw[maroon!70, thick] (37.5, -5.5) -- (16, -2);
        \draw[maroon!70, thick] (37.5, -5.5) -- (16, -4);
        \draw[fill=maroon!70, thick] (36+5, {-5-5}) rectangle (37+5, {-6-5});
        \draw[fill=maroon!70, thick] (25,-12) rectangle (28,-14);
        \draw[maroon!70, thick] (41.5, -10.5) -- (28, -12);
        \draw[maroon!70, thick] (41.5, -10.5) -- (28, -14);
\end{tikzpicture}}
\caption{The  entries of $\bA$ are independent with $A_{ij} \sim \mathcal{N}(0,\frac{1}{m/\R} W_{\sfr(i), \sfc(j)})$, where $\bW$ is the base matrix. Here $\bW$ is an $(\omega, \Lambda)$ base matrix with $\omega=3, \Lambda=7$ (see Definition \ref{def:ome_lamb_rho}). The white parts of $\bA$ and $\bW$ correspond to zeros.}
\label{fig:spatial_coupling_example}
\end{figure}
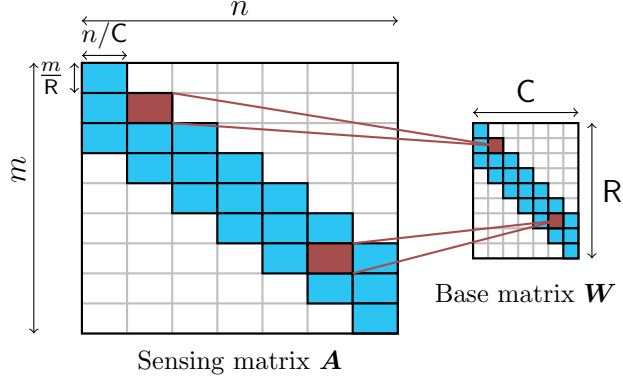

As in \cite{donoho2013information}, we will assume that entries of the base matrix $\bW$ are scaled to satisfy:%
\begin{equation}
    \sum_{\sfr=1}^{\Lr} W_{\sfr \sfc} = 1 \quad \text{ for } \sfc \in [\Lc], \qquad \kappa_1 \le \sum_{\sfc=1}^{\Lc} W_{\sfr \sfc} \le \kappa_2,
    \label{eq:Wrc_assumptions}
\end{equation}
for some $\kappa_1, \kappa_2 >0$.
The first condition in \eqref{eq:Wrc_assumptions} ensures that  the columns of the sensing matrix $\bA$ have expected squared norm equal to one, to match the standard assumption for AMP with i.i.d.~Gaussian matrices.   The second condition ensures that the variance of each entry of $\bz=\bA \bx$ is bounded above and below for all $i \in [m]$. Indeed, it can be verified that the $\sfr$th block of $\bz$,  for $\sfr \in [\R]$, has i.i.d. entries distributed according to the law of a zero-mean  Gaussian  random variable $Z_\sfr$ with variance
\begin{equation}
    \E\{ Z_{\sfr}^2 \} =  \frac{\E\{ X^2 \}}{\delin} \sum_{\sfc =1}^{\C} W_{\sfr\sfc} \in \left[ \frac{\E\{ X^2 \}}{\delin} \kappa_1, \, \frac{\E\{ X^2 \}}{\delin} \kappa_2 \right].
    \label{eq:Zr_sq}
\end{equation}

%
The trivial base matrix with $\Lr=\Lc=1$ (single entry equal to 1) corresponds to the sensing matrix with i.i.d.~$\mc{N}(0, \frac{1}{m})$ entries. The quantity, $\E\{ \| \bz \|^2 \}/m$, which can be interpreted as the measurement strength, is equal to $\E\{ X^2 \}/\delta$ for both the i.i.d.~Gaussian and spatially coupled matrices. Indeed, from \eqref{eq:Zr_sq}, we have that with spatial coupling,
\begin{equation}
    \frac{\E\{ \| \bz \|^2 \}}{m} = 
    \frac{\E\{ X^2 \}}{ \R \, \delin} \sum_{\sfr =1}^{\R} \sum_{\sfc =1}^{\C} W_{\sfr\sfc}
    = \frac{\E\{ X^2 \}}{\delta},
\label{eq:meas_strength}
\end{equation}
where the second equality follows from \eqref{eq:del_inner} and \eqref{eq:Wrc_assumptions}.
Thus the spatially coupled design can be fairly compared  to the i.i.d.~Gaussian one  as they have the same measurement strength for a given $\delta$.

To prove the Bayes optimality of spatial coupling with AMP,  we will use the following base matrix,  first used for sparse superposition coding over Gaussian channels \cite{rush2020capacity, hsieh2022GMAC}, and inspired by the construction of protograph based spatially coupled LDPC codes \cite{mitchell2015spatially}.

\begin{defi}
\label{def:ome_lamb_rho}
An $(\omega , \Lambda)$ base matrix $\bW$ is described by two parameters: the coupling width $\omega\geq1$ and the coupling length $\Lambda\geq 2\omega-1$. The matrix has $\Lr=\Lambda+\omega-1$ rows and $\Lc=\Lambda$ columns, with the $(\sfr,\sfc)$th entry of the base matrix, for  $\sfr \in [\Lr], \sfc\in[\Lc]$, given by
\begin{equation}
\label{eq:W_rc}
W_{\sfr \sfc} =
\begin{cases}
 	\ \frac{1}{\omega} \quad &\text{if} \ \sfc \leq \sfr \leq \sfc+\omega-1,\\
	\ 0 \quad &\text{otherwise}.
\end{cases}
\end{equation}
\end{defi}
The base matrix in Figure \ref{fig:spatial_coupling_example} has parameters $(\omega=3, \Lambda=7)$.

With an $(\omega, \Lambda)$ base matrix, from \eqref{eq:del_inner} we have $\delin= \frac{\Lambda}{\Lambda + \omega - 1}\delta $. As $\omega>1$ in spatially coupled systems,  the difference $\delta - \delin$ is often referred to as a ``rate loss'' in the literature of spatially coupled error correcting codes \cite{kudekar2011threshold, kudekar2013spatially, costello2014spatially}, and becomes negligible when $\Lambda$ is much larger than $\omega$.

\subsection{Spatially Coupled Generalized AMP} \label{subsec:scGAMP}

The Generalized AMP (GAMP) algorithm for estimation in a GLM with i.i.d.~Gaussian design was proposed by Rangan \cite{rangan2010generalized}. We now describe the AMP algorithm for estimating $\bx$ from $\by$ when $\bA$ is a spatially coupled Gaussian matrix. We call this algorithm Spatially Coupled Generalized AMP (SC-GAMP). For iteration $t \ge 1$, SC-GAMP iteratively computes $\bp(t) \in \reals^m$ and $\bq(t) \in \reals^n$ via  functions $\bargin(\cdot \, ; \, t): \reals^n \to \reals^n$ and $\bargout(\cdot, \cdot \,  ; \, t): \reals^m \times \reals^m \to \reals^m$.   Starting from an initialization $\bx(0) \equiv \bargin(\bq(0) \, ; \, 0 ) \in \reals^n$ and $\bp(0) = \bA \bx(0)$, for $t \ge 0$ it computes:
\begin{align}
      &  \bq(t+1)=\bargin(\bq(t) \, ; \, t )
      \, +  \, \bttq(t+1)\odot \bA^\top \bargout\left(\bp({t}),\by \, ; \, t\right),\label{eq:SC-GAMP1} \\
    & \bp(t+1)=\bA \bargin(\bq(t+1); \,  t+1)\, - \, \bttp(t+1) \odot \bargout(\bp(t),\by \, ; \, t ).
        \label{eq:SC-GAMP2}
\end{align}
Here the  vectors $\bttp(t+1) \in \reals^n$ and $\bttq(t+1) \in \reals^m$ are defined below via state evolution, and $\odot$ denotes the Hadamard (element-wise) product. The vector $\bargin(\bq(t) \, ; \, t )$ is the estimate of the signal at the end of iteration $t$.

The functions $\bargin$ and $\bargout$ are assumed to have a separable block-wise structure, defined as follows.  Partition the sets $[m]$ and $[n]$ into $\R$ and $\C$ equal-sized blocks (as shown in Figure \ref{fig:spatial_coupling_example}), respectively,
denoted by 
    \begin{equation}
       [n] = \cup_{\sfc=1}^{\Lc} \, \mc{J}_{\sfc}, \quad   [m] = \cup_{\sfr=1}^{\Lr} \, \mc{I}_{\sfr},
       \label{eq:nm_part}
    \end{equation}
where
\begin{align}
&\mc{J}_{\sfc} = \{(\sfc-1) \frac{n}{\C} +1, \ldots, \sfc \,  \frac{n}{\C} \},  \ \text{ for } \sfc \in [\C], \nonumber\\
&\mc{I}_{\sfr} = \{(\sfr-1) \frac{m}{\R} +1, \ldots, \sfr \,  \frac{m}{\R} \},  \ \text{ for } \sfr \in [\R].
\label{eq:part_def}
\end{align}
Then, for each $t \ge 0$, we assume that there exist functions $\gin(\cdot \, ; \, t+1): \reals\times [\C] \to \reals$ and $\gout( \cdot \,  ; \,  t): \reals^2 \times [\R] \to \reals$ such that 
\begin{align}
    & \bar{g}_{\text{in}, j} \left( \bq(t+1) \, ; \, t+1 \right)=
    \gin(q_j(t+1), \sfc \, ; \, t+1), \   \text{ for } j \in \mc{J}_{\sfc}, \nonumber \\
   &  \bar{g}_{\text{out}, i}\left(\bp(t),\by \, ; \, t\right) =
    \gout\left( p_i(t), y_i, \sfr \, ; \, t \right), \   \text{ for } i \in \mc{I}_{\sfr}. \label{eq:gingout_def}
\end{align}
In words, the functions $\bar{g}_{\text{in}} \left( \cdot \,  ; \, t+1 \right): \reals^{n} \to \reals^n$ and $\bar{g}_{\text{out}}\left(\cdot ,\cdot \, ; \, t\right): \reals^{m} \times \reals^m \to \reals^m$ act row-wise, with each component depending only on the index of the column block or row block.

\paragraph{State Evolution.}  We will show in Theorem \ref{thm:gen_sc_gamp_SE}  that under suitable assumptions, for each $t \ge 0$, the empirical distribution of $(\bq_{\sfc}(t) - \mu^{q}_\sfc(t) \,  \bx_{\sfc})$ converges to a Gaussian $\normal(0, \tq_{\sfc}(t))$, for each $\sfc \in [\C]$. Here $\bq_{\sfc}(t), \bx_{\sfc} \in \reals^{n/\C}$ denote the $\sfc$th block of $\bq(t), \bx \in \reals^n$, respectively. Thus the $\sfc$th block of the function $\bargin(; t)$ can be viewed as estimating $\bx_{\sfc}$ from an observation of the form $ \mu^{q}_\sfc(t) \,  \bx_{\sfc}$ plus  Gaussian noise of variance $\tq_{\sfc}(t)$. Analogously, the joint empirical distribution of the rows of $(\bp_\sfr(t), \, \bz_{\sfr}  ) \in \reals^{(m/\R) \times 2}$
converges to a bivariate Gaussian distribution $\normal(0, \Lambda_\sfr(t))$, 
for $\sfr \in [\R]$. Here $\bp_{\sfr}(t), \bz_{\sfr} \in \reals^{m/\R}$ denote the $\sfr$th block of $\bp(t), \bz \in \reals^m$, respectively. The constants  $\mu^{q}_\sfc(t), \tq_{\sfc}(t) \in \reals$, and $\Lambda_\sfr(t) \in \reals^{2 \times 2}$
are defined via the state evolution recursion below.

For $t \ge 1$, given the coefficients $\mu^{q}_\sfc(t), \tau^q_{\sfc}(t)$ for $\sfc \in [\C]$, and $\Lambda_{\sfr}(t)$ for $\sfr \in [\R]$,   define the random variables $Q_{\sfc}(t)$ and $P_{\sfr}(t)$ as follows:
\begin{align} 
    & Q_{\sfc}(t)  =  \mu^{q}_\sfc(t) X \, +  \, G^q_{\sfc}(t), \quad  \text{ where } X \sim P_X \, \perp \,  
    G^q_{\sfc}(t) \sim \normal(0, \tau^q_{\sfc}(t)), \nonumber\\
    & (P_{\sfr}(t),  \, Z_\sfr) \, \sim \,  \normal(0, \Lambda_\sfr(t)).\label{eq:PtQt_def}
\end{align}
Recalling that $\by = \varphi(\bz, \bveps)$ where $\bz = \bA \bx$, we write $\gout(p, y, \sfr \, ; t) =
\gout(p, \varphi(z, \veps), \sfr \, ; \, t)$ and let 
$ \goutpr(p, \varphi(z, \veps), \sfr \, ; \, t)$ denote the derivative with respect to  the first argument and
$\partial_z \, \gout(p, \varphi(z, \veps), \sfr \, ; \, t)$ the derivative with respect to $z$. Then the coefficients for $(t+1)$ are recursively  computed as follows, for $\sfc \in [\C]$ and $\sfr \in [\R]$: 
\begin{align}
    & \tq_{\sfc}(t+1) = \frac{\sum_{\sfr =1}^{\R} W_{\sfr \sfc} \, \E\{ \gout( P_{\sfr}(t), \, \varphi(Z_{\sfr}, \bar{\veps}), \sfr \,  ; \, t )^2 \}}{ \big[ \sum_{\sfr =1}^{\R} W_{\sfr \sfc} \, \E\{ \goutpr( P_{\sfr}(t), \, \varphi(Z_{\sfr}, \bar{\veps}), \sfr \,  ; \, t ) \} \big]^2}, \label{eq:tq_t1_def} \\
   &  \mu_{\sfc}^q(t+1) = \frac{ \sum_{\sfr =1}^{\R} W_{\sfr \sfc} \, \E\{ \partial_z \gout( P_{\sfr}(t), \, \varphi(Z_{\sfr}, \bar{\veps}), \sfr \,  ; \, t ) \}}{- \sum_{\sfr =1}^{\R} W_{\sfr \sfc} \, \E\{ \goutpr( P_{\sfr}(t), \, \varphi(Z_{\sfr}, \bar{\veps}), \sfr \,  ; \, t ) \}  }, \label{eq:muq_t1_def} \\
  & [\Lambda_\sfr(t+1)]_{11} =   \frac{1}{\delin} \sum_{\sfc =1}^{\C} W_{\sfr\sfc} 
  \E\{ \gin(Q_{\sfc}(t+1),  \sfc \, ;  \, (t+1))^2   \}, \,  
   \label{eq:Lambda11} \\
 & [\Lambda_\sfr(t+1)]_{12} =  [\Lambda_\sfr(t+1)]_{21} = \frac{1}{\delin} \sum_{\sfc =1}^{\C} W_{\sfr\sfc} \E\{ X \, \gin(Q_{\sfc}(t+1), \sfc \, ;  \, (t+1))\}, \, \label{eq:Lambda12} \\
    & [\Lambda_\sfr(t+1)]_{22} = \E\{ Z_{\sfr}^2 \} = \frac{\E\{ X^2 \}}{\delin} \sum_{\sfc =1}^{\C} W_{\sfr\sfc} \, . \label{eq:Lambda22}
\end{align}

To initialize the state evolution, we make the following assumption on the SC-GAMP initialization $\bx(0) \in \reals^n$. 

\textbf{(A0)} Denoting by $\bx(0)_\sfc \in \reals^{n/\C}$ the $\sfc$th block of $\bx(0) \in  \reals^n$, 
we assume that there exists a symmetric non-negative definite $\Xi_\sfc \in \reals^{2 \times 2}$ for each $\sfc \in [\C]$ such that  we almost surely have
\begin{equation}
\lim_{n \to \infty}  \,    \frac{1}{n/\C} 
\begin{bmatrix}
\< \bx(0)_\sfc , \bx(0)_\sfc \> & \< \bx_\sfc , \bx(0)_\sfc \> \\
\< \bx_\sfc , \bx(0)_\sfc \> & \< \bx_\sfc , \bx_\sfc \> 
\end{bmatrix} = \Xi_\sfc.
\label{eq:init_limits}
\end{equation}
From the assumptions on signal in Section \ref{subsec:model_assump}, we have that $[\Xi_\sfc]_{22} = \E\{ X^2\}$.

The state evolution iteration is initialized with 
$\Lambda_{\sfr}(0)$ for $\sfr \in [\R]$, whose entries are given by:
\begin{align}
& [\Lambda_{\sfr}(0)]_{11}=  \frac{1}{\delin} \sum_{\sfc =1}^{\C} W_{\sfr\sfc} \left[ \Xi_\sfc \right]_{11}, \nonumber \\ 
& [\Lambda_{\sfr}(0) ]_{22} = \frac{1}{\delin} \sum_{\sfc =1}^{\C} W_{\sfr\sfc} \left[ \Xi_\sfc \right]_{22}= \E\{ Z_{\sfr}^2 \}, \nonumber \\
& [\Lambda_{\sfr}(0)]_{12}= [\Lambda_{\sfr}(0)]_{21}
= \frac{1}{\delin} \sum_{\sfc =1}^{\C} W_{\sfr\sfc} \left[ \Xi_\sfc \right]_{12}.
\label{eq:Lambda0_def}
\end{align}

The entries of the coefficient vectors $\bttp(t+1) \in \reals^m$ and $\bttq(t+1) \in \reals^{n}$ in \eqref{eq:SC-GAMP1}-\eqref{eq:SC-GAMP2} are defined block-wise as follows. Starting with $\alpha^p_{\sfr}(0)=0$ for all $\sfr$, for $t \ge 0$, we  recursively compute the following for $\sfc \in [\C ]$ and $\sfr \in [\R]$:
\begin{align}
    & \alpha^q_{\sfc}(t+1)   = -\left( \sum_{\sfr=1}^\R W_{\sfr\sfc} \E\{ \goutpr(P_{\sfr}(t), Y_\sfr, \sfr \, ; \, t) \} \right)^{-1}, 
     \label{eq:alpha_qtq_def} \\
     & \alpha^p_{\sfr}(t+1)  =   \frac{1}{\delin}\sum_{\sfc=1}^\C 
       W_{\sfr\sfc} \, \alpha^q_\sfc(t+1) \,
       \E\{ \ginpr\left( Q_\sfc(t+1), \sfc \, ; \, t+1 \right)\}. \label{eq:alpha_pt_def}
\end{align} 
We then set 
\begin{align}
   &\alpha^p_{i}(t+1) = \alpha^p_{\sfr}(t+1), \ \  \text{ for } i \in \mc{I}_{\sfr},\nonumber\\ 
   &\alpha^q_{j}(t+1) = \alpha^q_{\sfc}(t+1), \ \ \text{ for } j \in \mc{J}_{\sfc}.
\end{align}

We note that for $(P_\sfr(t), \, Z_{\sfr} ) \sim \normal(0, \Lambda_{\sfr}(t))$, by standard properties of Gaussian random vectors we have:
\begin{equation}
    (P_\sfr(t), \,  Z_{\sfr} )  \, \disteq \,
    (  P_{\sfr}(t), \, \mu_{\sfr}^p(t) P_{\sfr}(t)  +  G^p_\sfr(t) ),
    \label{eq:PZ_dist_alt}
\end{equation}
where  $G^p_\sfr(t) \sim \normal(0, \tp_{\sfr}(t))$ is independent of $P_{\sfr}(t)$ with
\begin{align}
    \mu_{\sfr}^p(t) = \frac{[\Lambda_{\sfr}(t)]_{12}}{[\Lambda_{\sfr}(t)]_{11}}, \qquad
    \tp_{\sfr}(t) = [\Lambda_{\sfr}(t)]_{22} \, - \,  \frac{[\Lambda_{\sfr}(t)]_{12}^2}{[\Lambda_{\sfr}(t)]_{11}}.
    \label{eq:mup_t1_def}
\end{align}

\paragraph{Performance Characterization of SC-GAMP.}
 For the state evolution result  (Theorem \ref{thm:gen_sc_gamp_SE} below), we make the following assumption on the functions $\gin, \gout$ in \eqref{eq:gingout_def} used to define SC-GAMP.

\textbf{(A1)}  For $t \ge 0$, the function $\gin( \cdot, \, \sfc \, ; t+1 )$ is Lipschitz on $\reals$ for $\sfc \in [\C]$, and the function $\tilde{g}_{\sfr, t}: (p, z, \veps) \mapsto \gout(p, \, \varphi(z, \veps), \, \sfr; \, t)$ is Lipschitz on $\reals^{3}$ for $\sfr \in [\R]$.

 The state evolution result is stated in terms of \emph{pseudo-Lipschitz} test functions. A function $\Psi:\reals^d\rightarrow\reals$ is called pseudo-Lipschitz of order $k$ if $|\Psi(x)-\Psi(y)|\leq C (1+\|x\|_2^{k-1}+\|y\|_2^{k-1})\|x-y\|_2$ for all $x,y\in\reals^d$. Pseudo-Lipschitz functions of order $k=2$ include $\Psi(u)=u^2$ and $\Psi(u,v)=uv$.

\begin{theorem} Consider a GLM  with a spatially coupled sensing matrix defined via a base matrix $\bW$ \eqref{eq:Wrc_assumptions},
and the SC-GAMP algorithm in \eqref{eq:SC-GAMP1}-\eqref{eq:SC-GAMP2}. Assume that the model assumptions
in Section \ref{subsec:model_assump} and the SC-GAMP assumptions in \textbf{(A0)},\textbf{(A1)} hold. 
Let $\Psi: \reals^2 \to \reals$ and $\Phi: \reals^3 \to \reals$ be any pseudo-Lipschitz functions of order $k$, where $k$ is specified in Section \ref{subsec:model_assump}. Then, for all $t \ge 1$ and each $\sfr \in [\Lr]$ and $\sfc \in [\Lc]$, we almost surely have:
\begin{align}
 &  \lim_{n \to \infty} \,  \frac{1}{n/\C} \sum_{j \in  \mc{J}_\sfc} \Psi(q_j(t) \, , x_j)  = \E\{ \Psi( Q_{\sfc}(t) \, , X) \},     \label{eq:Qt_result} \\
 &  \lim_{m \to \infty} \,  \frac{1}{m/\R} \sum_{i \in \mc{I}_\sfr } \Phi(p_i(t) \, , z_i \, , \veps_i)  = \E\{ \Phi( P_{\sfr}(t) \, , Z_{\sfr}, \, \bar{\veps} ) \}.
      \label{eq:Pt_result}
\end{align}
\label{thm:gen_sc_gamp_SE}
The laws of the random variables in \eqref{eq:Qt_result} and \eqref{eq:Pt_result} are given by \eqref{eq:PtQt_def}.
\end{theorem}

The proof of the theorem is given in Section \ref{sec:proof_gen_sc_gamp_SE}. The main idea is to show, via a suitable change of variables, that the spatially coupled GAMP in \eqref{eq:SC-GAMP1}-\eqref{eq:SC-GAMP2} is an instance of an abstract AMP iteration for i.i.d.~Gaussian matrices with \emph{matrix-valued} iterates. A state evolution characterization for the abstract AMP iteration can be obtained using the result in  \cite{javanmard2013state}. This result is translated via the change of variables to establish the state evolution  characterization in Theorem \ref{thm:gen_sc_gamp_SE}.

Theorem \ref{thm:gen_sc_gamp_SE} allows us to compute the performance measures such as asymptotic MSE of SC-GAMP. Indeed, in \eqref{eq:Qt_result} taking $\Psi(q,x) = (x- \gin(q, \sfc ; t) )^2$ for $\sfc \in [\C]$,  we obtain:
\begin{equation}
\begin{split}
   &  \lim_{n \to \infty} \,  \frac{\| \bx - \bargin(\bq(t))\|^2}{n}  =  \frac{1}{\C} \sum_{\sfc=1}^{\C} \, \E\{ (X - \gin(Q_{\sfc}(t), \sfc ; t ) )^2\} \ \text{ a.s.}
    \label{eq:SC-GAMP-MSE}
    \end{split}
\end{equation}

\section{MMSE and the performance  of Bayes SC-GAMP} \label{sec:mmse-fp-bayes}

The limiting value of $\MMSE_n$ in  \eqref{eq:MMSE_def} for an i.i.d.~Gaussian sensing matrix is characterized in terms of a \emph{potential function} \cite{barbier2019optimal}.
 \begin{defi}[Potential function]\label{def:pot}
For $x \in [0, \var(X)]$, let 
\begin{equation}
    \ups(x;  \delta) :=1 - \frac{\delta}{x} \, \E_{P,Y}  \left\{\var\left(Z\middle|P, Y; \, \frac{x}{\delta} \right)\right\}, 
    \label{eq:ups_x}
\end{equation}
where 
$Y=\varphi(Z, \, \bar{\veps})$ with $\bar{\veps}\sim P_{\bar{\veps}} \perp Z$, and $\var( Z | P, Y; \, \frac{x}{\delta})$ denotes the conditional variance computed with  $(P, Z) \sim \normal(0, \Lambda)$ where
\begin{equation}
   \Lambda = \frac{1}{\delta}\begin{bmatrix}
   \E\{X^2 \} -x  &  \E\{X^2\} -x \\
   \E\{X^2\} -x & \E\{X^2\}
   \end{bmatrix}.
   \label{eq:ups_x_def}
\end{equation}
 Then the scalar potential function $\pot(x; \delta)$
  is defined as
 \begin{equation} \label{eq:pot}
 \begin{split}
    &  U(x;  \delta) :=  - {\delta} \ups(x;  \delta)   \, +  \, \int_0^x \frac{\delta}{z} \ups(z;  \delta) \,  dz  \,  +  \, 2I\left(X \, ;  \sqrt{( \delta/x) \ups(x; \delta)} \,  X + G_0 \right),
\end{split}
 \end{equation} 
 for $x \in [0, \var(X)]$, where the mutual information is computed with $X \sim P_X \, \perp G_0 \sim \normal(0,1)$.
 \end{defi}

Potential functions are widely used in statistical physics to characterize the limiting free energy (or equivalently, the  mutual information) in high-dimensional estimation problems; see, e.g.,  \cite{zdeborova2016statistical}.  We will not use the statistical physics interpretation here, and use the potential function only to characterize the fixed points of the state evolution recursion, both with and without spatial coupling (see Theorem \ref{thm:SC_FPs}). Moreover, the following  result from \cite{barbier2019optimal} shows that the limiting MMSE with an i.i.d.~Gaussian sensing matrix is given by the minimum of the potential function.

 \begin{theorem}[MMSE for i.i.d.~Gaussian design \cite{barbier2019optimal}] Consider a GLM with sensing matrix $\Aiid$ whose entries are i.i.d.~$\sim \normal(0, \frac{1}{m})$. Suppose that the components of $\bx$ and $\bveps$ are i.i.d.~according to $P_X$ and $P_{\bar{\veps}}$, respectively, and $\frac{m}{n} \to \delta$ as $n \to \infty$. Furthermore, assume that one of the conditions \eqref{eq:sign_anti_sym} or \eqref{eq:exp_ne_0} holds, and that $U(x;  \delta)$ has  a unique minimum. Then,
 \begin{equation}
     \lim_{n \to \infty}  \frac{1}{n} \E\{ \| \bx - \E\{\bx \mid \Aiid, \, \by \} \|^2\} = \argmin_{ x \in [0, \var(X)]} \, U(x; \delta).
 \end{equation}
 \label{thm:iidMMSE}
 \end{theorem}
   The assumption that one of the conditions \eqref{eq:sign_anti_sym} or \eqref{eq:exp_ne_0} holds is to avoid issues of sign-symmetry in the model, as discussed in Section \ref{subsec:model_assump}. Another way to avoid the sign-symmetry issue without further assumptions  is by stating the result  in terms of the MMSE of estimating the rank-one matrix $\bx \bx^{\sT}$; see \cite[Theorem 2]{barbier2019optimal}. The result in \cite{barbier2019optimal} is expressed in terms of a maximization over a different potential function, which is equivalent to $U(x; \delta)$ up to a change of sign and additive constants. The alternative potential function formulation  is described in Appendix \ref{app:pot_fun_equiv}, where we also show its equivalence to the one in Definition \ref{def:pot}.
We mention that Theorem \ref{thm:iidMMSE} requires a few additional technical conditions \cite[Sec. 4]{barbier2019optimal}, which are listed in Appendix \ref{app:pot_fun_equiv}.

The significance of Theorem \ref{thm:iidMMSE} is illustrated in Figure \ref{fig:pot-pr}, which plots the potential function $\pot(x;\delta)$ for noiseless phase retrieval as a function of  $x\in[0,\var(X)]$, for several sampling ratios $\delta$.   From Theorem \ref{thm:iidMMSE}, for each $\delta$, the global minimizer of the potential function gives the asymptotic MMSE  with an i.i.d.~Gaussian design. For $\delta=0.2$, the global minimizer is close to 1, which implies it is impossible to meaningfully estimate the signal $\bx$. For $\delta=0.27$, there are two global minimizers of the potential function. For $\delta>0.27$, we observe that the global minimizer lies at 0, so the Bayes-optimal estimator asymptotically achieves perfect reconstruction for these sampling ratios. Details of how the potential function and its minimizers are computed can be found in Appendix \ref{app:pot}.

\begin{figure}[t]
		\centering
		\includegraphics[width=0.6\textwidth]{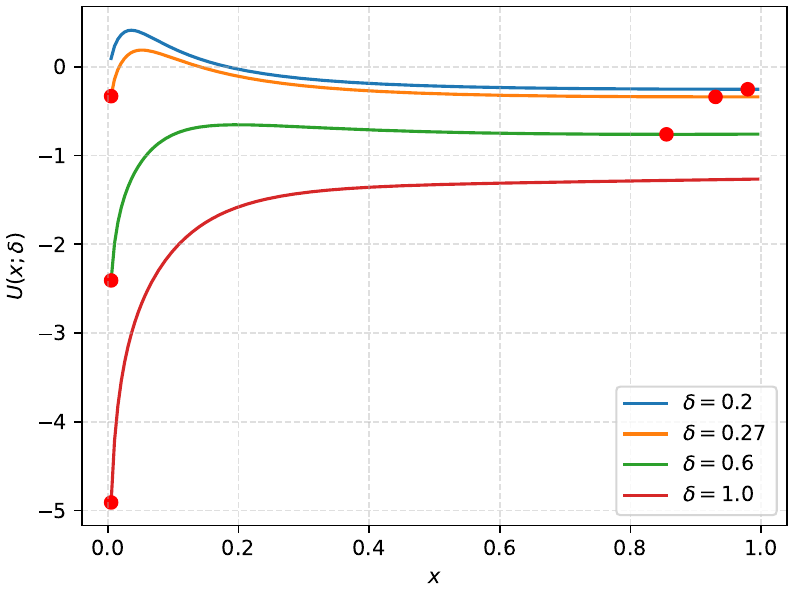}
		\caption{Potential function $\pot(x;\delta)$ for noiseless phase retrieval, for different values of sampling ratio $\delta$ and $x\in[0,\var(X)]$. The signal prior is $P_X(a)=1- P_X(-a)=0.6$, with $a$ chosen so that $\var(X)=1$. The red dots represent the largest stationary point and global minimizer of each function.}
		\label{fig:pot-pr}
  \vspace{-7pt}
\end{figure}

\subsection{Bayes  SC-GAMP}
Returning to the spatially coupled setting,  consider SC-GAMP with the following choice of functions:
\begin{align}
    & \ginbayes(q, \sfc \, ; \, t) = \E\{ X \mid Q_{\sfc}(t) =q \}, \ \text{ for } \sfc \in [\C], \label{eq:ginbayes}\\
    &  \goutbayes(p, y, \sfr \,  ; \, t) = \frac{1}{\tp_{\sfr}(t)}\left( \E\{ Z_{\sfr} \mid P_{\sfr}(t) =p, \, Y_{\sfr}=y \} \, - \, p \right), \  \text{ for } \sfr \in [\R]. \label{eq:goutbayes}
\end{align}
For the uncoupled case, i.e., $\R = \C =1$, these choices maximize the correlations $\frac{(\E\{P(t)Z\})^2} {\E\{ P(t)^2 \} \E\{ Z^2 \}}$ and $\frac{(\E\{ X Q(t) \})^2}{\E\{X^2\} \E\{ Q(t)^2 \}}$ for each $t \ge 1$  \cite[Sec. 4.2]{fengAMP22}. We therefore refer to the algorithm with this choice of functions as Bayes SC-GAMP. We remark that the choices in \eqref{eq:ginbayes}-\eqref{eq:goutbayes} do not necessarily optimize the correlations for the coupled case, i.e., $\frac{(\E\{P_\sfr(t)Z_\sfr\})^2} {\E\{ P_\sfr(t)^2 \} \E\{ Z_\sfr^2 \}}$ and $\frac{(\E\{ X Q_\sfc(t) \})^2}{\E\{X^2\} \E\{ Q_\sfc(t)^2 \}}$ for $\sfr \in [\R]$ and $\sfc \in [\C]$.

Our second result shows that with an $(\omega, \Lambda)$ base matrix (see Definition \ref{def:ome_lamb_rho}), the MSE of SC-GAMP after a large number of iterations is bounded in terms  of the minimizer of  $U(x ; \delin)$, where $\delin = \delta \frac{\Lambda}{\Lambda+\omega -1}$. Moreover, with an i.i.d.~Gaussian sensing matrix, the MSE of GAMP is characterized by the largest stationary point of  
$U(x ; \delta)$ in $x \in [0, \var(X)]$. 
\begin{theorem}
    Consider a GLM with the assumptions in Section \ref{subsec:model_assump}, and estimation via Bayes SC-GAMP initialized with 
    $\bx(0) = \E\{ X\} \boldsymbol{1}_n$. Also assume that  \textbf{(A1)} is satisfied by  the functions $\ginbayes$ and $\goutbayes$, and recall that the vector-valued functions $\bargin^*, \bargout^*$ are defined according to \eqref{eq:gingout_def}. Then:

    1) With a spatially coupled sensing matrix defined via an $(\om, \lam)$ base matrix, for any $\gamma>0$ there exist $\omega_0<\infty$ and $t_0 < \infty$ such that for all $\omega>\omega_0$ and $t > t_0$, the asymptotic MSE of Bayes SC-GAMP almost surely  satisfies:
    \begin{align}
    & \lim_{n\to\infty}\frac{1}{n}\| \bx \, - \, \bargin^*(\bq(t); \, t ) \|^2 
  \leq  \left( \max\left(\argmin_{x\in [0, \var(X)]} U(x; \delin)\right)+\gamma\right) \frac{\Lambda+\omega}{\Lambda}.
    \label{eq:omLam_MSE}
\end{align}
2) With an i.i.d.~Gaussian sensing matrix (i.e., $1 \times 1$ base matrix with $W_{11}=1$), the asymptotic MSE of Bayes GAMP converges almost surely as:
\begin{equation}
\begin{split}
    & \lim_{t\to\infty}\lim_{n\to\infty}\frac{1}{n}\| \bx \, - \, \bargin^*(\bq(t); \, t )\|^2    = \max\left\{ x \in [0, \var(X) \}] \, : \, \frac{\partial U(x; \delta)}{\partial x}  =0 \right\}.
\end{split} \label{eq:SC_FPs_iid}
\end{equation}
\label{thm:SC_FPs}
\end{theorem}
The proof of the theorem is given in Section \ref{subsec:proof_SC_FPs}.

If $U(x; \delin)$ has  multiple  minimizers in $x$ (e.g., the $\delta=0.27$ curve in  Figure \ref{fig:pot-pr}), 
then  according to \eqref{eq:omLam_MSE}, the asymptotic MSE of SC-GAMP will be upper bounded in terms of the largest of these  minimizers. If $U(x; \delin)$ has a unique minimizer, then we have the following corollary.

\begin{corr}[Bayes-optimality of SC-GAMP]
    Consider the setup of Theorem \ref{thm:SC_FPs}, part 1, and suppose that there exists some $\delta_0 > \delta$ such that: i) $U(x; \delin)$ has a unique minimizer in $x$ for $\delin \in [\delta, \delta_0]$, and ii) the minimizer in $x$ of $U(x; \delin)$ is continuous in  $\delta$ for $\delin \in [\delta, \delta_0]$. Then for any $\epsilon>0$, there exist finite $\omega_0, t_0 $ such that for all $\omega > \omega_0$, $t > t_0$, and $\Lambda$ sufficiently large, the asymptotic MSE of Bayes SC-GAMP almost surely satisfies:
    \begin{equation}
    \begin{split}
          \lim_{n\to\infty}\frac{1}{n}\| \bx \, - \, \bargin^*(\bq(t); \, t ) \|^2  
         & \le 
            \lim_{n \to \infty}  \frac{1}{n} \E\{ \| \bx - \E\{\bx \mid \Aiid, \, \by \} \|^2\}  + \epsilon \\
            & = \argmin_{ x \in [0, \var(X)]} \, U(x; \delta) + \epsilon.
    \end{split}
        \label{eq:SC_MMSE_comp}
    \end{equation}    
    Here $\Aiid$ is a sensing matrix whose entries are i.i.d.~$\sim \normal(0, \frac{1}{m})$.
    \label{corr:BayesOpt}
\end{corr}
\begin{proof}
 Since the minimizer in $x$ of $U(x; \delin)$ is unique and a continuous function of $\delin$ for $\delin \in [\delta, \delta_0]$, we have that 
$\argmin_{x\in [0, \var(X)]} U(x; \delin) \to  \argmin_{x\in [0, \var(X)]} U(x; \delta)$ as $\delin \to \delta$ from above. 
The inequality in \eqref{eq:SC_MMSE_comp} then follows by combining  Part 1  of Theorem \ref{thm:SC_FPs}  with the statement of Theorem \ref{thm:iidMMSE}, and noting that $\delin = \delta \frac{\Lambda}{\Lambda+\omega-1}$ can be made arbitrarily close to $\delta$ for sufficiently large $\frac{\Lambda}{\omega}$.
\end{proof}

The Bayes-optimality of SC-GAMP  is an instance of 
\emph{threshold saturation}, the phenomenon where the error performance of message passing in a spatially coupled system matches the optimal error performance in the corresponding uncoupled system, The threshold saturation phenomenon has been previously established for LDPC codes, compressed sensing, and sparse superposition codes \cite{kudekar2011threshold,  donoho2013information, yedla2014asimple, kumar2014threshold, kudekar2015wave, barbier2016proof, hsieh2022GMAC}.

If $U(x; \delta)$ has a unique stationary point in $x \in [0, \var(X)]$ at which the minimum is attained, then  Theorem \ref{thm:SC_FPs}  and Corollary \ref{corr:BayesOpt} imply that the MMSE can be achieved by an i.i.d.~Gaussian sensing matrix with estimation via standard GAMP, i.e., spatial coupling is not required in this case. However,  when $U(x; \delta)$  has multiple stationary points (e.g., the $\delta=0.6$ curve in Figure \ref{fig:pot-pr}), the MSE of  a spatially coupled design can be substantially lower than that of an i.i.d.~design. This is illustrated by the numerical results in Figure \ref{fig:sc-gamp-pr} and \ref{fig:sc-gamp-relu}, where there is a significant gap between the MSE of the two designs, in the range $0.5 < \delta \leq 0.8$ for noiseless phase retrieval and  $0.6 < \delta \leq 1$ for rectified linear regression.

\section{Numerical Results} \label{sec:numerical_results}

We evaluate the empirical performance of Bayes SC-GAMP for two different GLMs:  noiseless phase retrieval ($\by=|\bA\bx|^2$) and noiseless rectified linear regression ($\by=\max\left(\bA\bx,\boldsymbol{0}\right)$), shown in Figure \ref{fig:sc-gamp-pr} and Figure \ref{fig:sc-gamp-relu} respectively. 
The signal entries were sampled i.i.d.~from a two-point prior for phase retrieval, and from a sparse three-point prior  for rectified linear regression. 
An estimate of the  signal was obtained using the SC-GAMP algorithm defined in \eqref{eq:SC-GAMP1}-\eqref{eq:SC-GAMP2}, with the Bayes-optimal choices $\ginbayes$ and $\goutbayes$ for each model, as defined in \eqref{eq:ginbayes}-\eqref{eq:goutbayes}. The expressions for $\ginbayes$ and $\goutbayes$ in each case, as well as details of evaluating the potential function and its stationary points, are given in Appendix \ref{app:comp_GLM}. The SC-GAMP algorithm was run for a maximum of 300 iterations, and stopped early 
if the relative difference between the estimated SE parameter $\boldsymbol{\htp}(t)$ and its previous iterate $\boldsymbol{\htp}(t-1)$ fell below a prescribed threshold. 
Details on how these SE parameters can be estimated at runtime can be found in Appendix \ref{app:comp_SE}.

\begin{figure}[!t]
		\centering
		\includegraphics[width=0.73\textwidth]{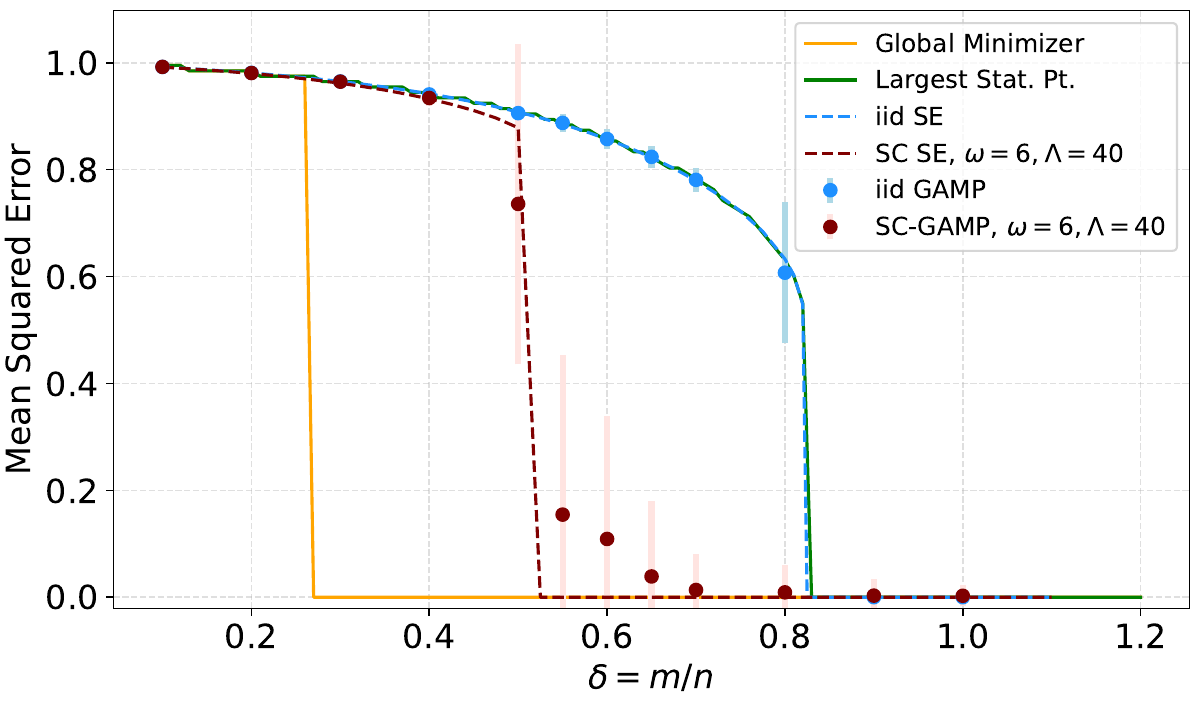} 	
		\caption{MSE of Bayes-optimal  GAMP for noiseless phase retrieval, with i.i.d.~Gaussian and spatially coupled designs. The signal entries are drawn  i.i.d.~from $\{-a,a\}$ with probabilities $\{0.4,0.6\}$, and $a$ chosen such that the variance is $1$. Signal dimension $n=20000$, the base matrix parameters are $(\omega=6,\Lambda=40)$, and the empirical performance is obtained by averaging over 100 trials and the error bars represent $\pm1$ standard deviation.}
		\label{fig:sc-gamp-pr}
\end{figure}

\begin{figure}[!t]
    \centering
    \includegraphics[width=0.73\textwidth]{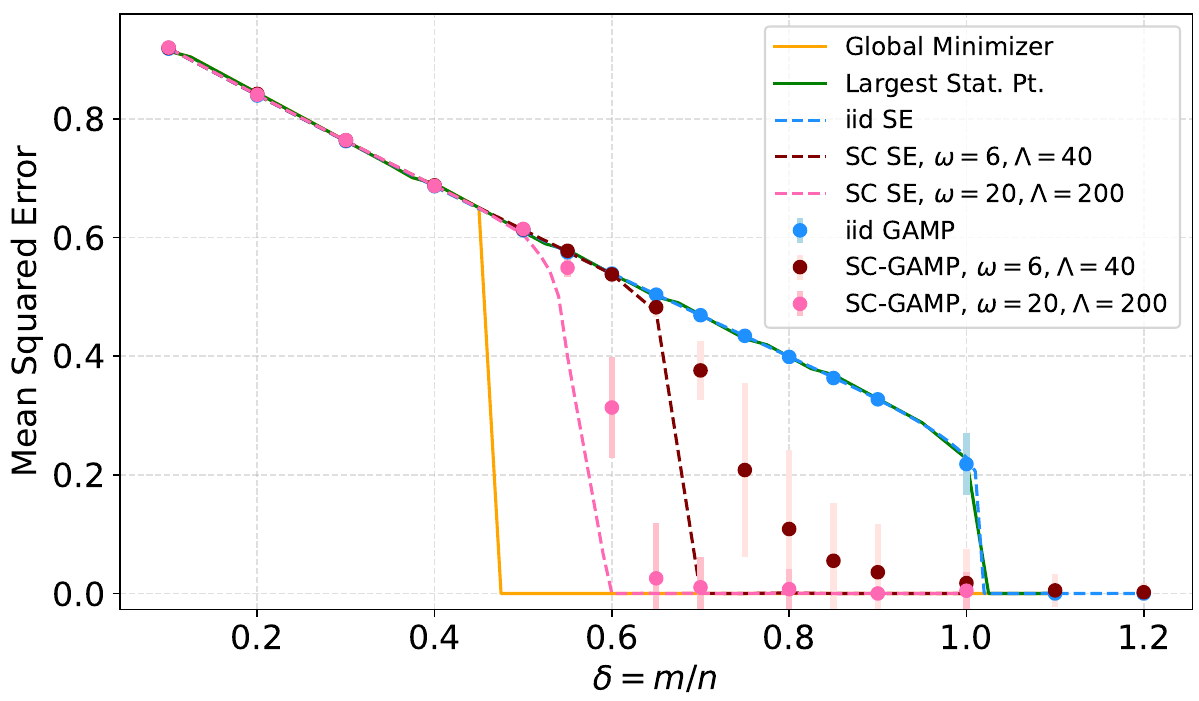}
    \caption{MSE of Bayes-optimal  GAMP for noiseless rectified linear regression, with i.i.d.~Gaussian and 2 different spatially coupled designs. The signal entries are drawn i.i.d.~from $\{-b,0,b\}$ with probabilities $\{0.25,0.5,0.25\}$, and $b$ chosen such that the variance is $1$. Signal dimension $n=20000$ for i.i.d.~GAMP. For SC-GAMP, the block dimension $n/\Lambda$ is fixed to $500$, and two  sets of base matrix parameters are considered: $(\omega=6,\Lambda=40)$ and $(\omega=20,\Lambda=200)$. The empirical performance is obtained by averaging over 100 trials and the error bars represent $\pm1$ standard deviation.}
    \label{fig:sc-gamp-relu}
\end{figure}

Figures \ref{fig:sc-gamp-pr} and \ref{fig:sc-gamp-relu} both plot the MSE of Bayes SC-GAMP and Bayes GAMP (the latter with an i.i.d.~Gaussian matrix) at convergence for a range of $\delta$ values. 
The orange solid line is the global minimizer of $U(x; \delta)$ and the green solid line the largest stationary point of $U(x; \delta)$.
The dashed lines are the state evolution MSE predictions for GAMP, given by \eqref{eq:SC-GAMP-MSE}. Both figures show that the average MSE for an i.i.d.~Gaussian sensing matrix (solid blue) closely matches the SE prediction (dashed blue), and the largest minimizer (green) curves, as expected from \eqref{eq:SC_FPs_iid}.
For both phase retrieval and rectified linear regression,  the  curves  show that the minimum $\delta$ required  for near-perfect reconstruction is significantly lower for the spatially coupled design compared to the i.i.d.~Gaussian one.

The difference between the empirical performance of SC-GAMP and its state evolution curve is due to finite length effects. In both figures, the difference between the $(\omega=6,\Lambda=40)$ state evolution curve (dashed dark red) and the global minimizer (orange) is because of the small values of  $\omega$ and $\Lambda$. The effect of increasing these values is illustrated by the dashed and solid pink lines in Figure \ref{fig:sc-gamp-relu}, which show that increasing the base matrix dimensions to $(\omega=20,\Lambda=200)$ noticeably reduces the minimum $\delta$ required by SC-GAMP for near-perfect reconstruction, approaching the performance of the MMSE estimator. We do not expect that the $\omega > \omega_0$ condition of Corollary \ref{corr:BayesOpt} is satisfied with  $\omega=6$, and even for $\omega=20$ it appears that $\Lambda=200$ is not sufficiently large.

For the $(\omega=20, \Lambda=200)$ setting in Figure \ref{fig:sc-gamp-relu}, the signal dimension is $n=100000$, leading to a design matrix with $\sim 10^{10}$ entries. Since it is challenging to store such a large design Gaussian design matrix in memory, a Discrete Cosine Transform (DCT) based implementation was used for the SC-GAMP in this case. Details of this DCT implementation can be found in Appendix \ref{app:dct}.

\section{Proofs of main results} \label{sec:proofs}

\subsection{Proof of Theorem \ref{thm:gen_sc_gamp_SE}}
\label{sec:proof_gen_sc_gamp_SE}

We first describe in Section \ref{subsec:gen_rect_AMP} an abstract AMP iteration  with matrix-valued iterates.  The state evolution result for this  abstract AMP   was established  in \cite{javanmard2013state}; see also \cite[Sec. 6.7]{fengAMP22}, \cite{Ger21}. We use this result to prove Theorem \ref{thm:gen_sc_gamp_SE} in Section \ref{subsec:proof_main}, by showing that the spatially coupled GAMP in \eqref{eq:SC-GAMP1}-\eqref{eq:SC-GAMP2} is an instance of an abstract AMP iteration with matrix-valued iterates for i.i.d.~Gaussian matrices.

\subsubsection{Abstract  AMP recursion for i.i.d.~Gaussian matrices} \label{subsec:gen_rect_AMP}

Let $\bM \in \reals^{m \times n}$ be a random matrix with entries $M_{ij} \sim_{\text{iid}} \normal(0, \frac{1}{m})$. Consider the following recursion with iterates $\be^{t+1} \in \reals^{m \times \ell_E}$ and $\bh^{t+1} \in \reals^{n \times \ell_H}$, defined recursively as follows. For $t \ge 0$:
\begin{align}
    \bh^{t+1} & = \bM^{\sT} g_t(\be^{t}, \, \bgamma)
    \, - \, f_t(\bh^t, \, \bbeta) \sD^{\sT}_t \, ,\label{eq:Ht_update} \\
    \be^{t+1} &  = \bM f_{t+1}(\bh^{t+1}, \, \bbeta) \, - \, g_{t}(\be^{t}, \, \bgamma) \sB_{t+1}^{\sT}. \label{eq:Et_update} 
 \end{align}
 The iteration is initialized with some $\bi^0 \equiv  f_0(\bh^0, \, \bbeta) \in \reals^{n \times \ell_E}$ and $\be^0 = \bM \bi^0$.
Here, the functions $f_{t+1}: \reals^{n \times (\ell_H +1)} \to \reals^{n \times \ell_E}$ and $g_t: \reals^{m \times (\ell_E +1)} \to \reals^{n \times \ell_H}$ are defined component-wise as follows, for $t \ge 0$. 
Recalling the partition of the sets $[m]$ and $[n]$ in \eqref{eq:nm_part}-\eqref{eq:part_def}, 
 we assume there exist functions $\tf_{t+1} : \reals^{\ell_H} \times \reals \times [\Lc] \to \reals^{\ell_E}$ and 
$\tg_t: \reals^{\ell_E} \times \reals \times [\Lr] \to \reals^{\ell_H}$ such that 
\begin{equation}
    \begin{split}
            & f_{t+1,j}(\bh, \bbeta) = \tf_{t+1}(h_j, \, \beta_j, \ \sfc ),  \quad  \text{ for } j \in \mc{J}_{\sfc}, \quad \bh \in \reals^{n \times \ell_H}, \ \bbeta \in \reals^n \\
            & g_{t,i}(\be, \bgamma) = \tg_{t}(e_i, \, \gamma_i, \ \sfr ), \quad \text{ for } i \in \mc{I}_{\sfr}, \quad \be \in \reals^{m \times \ell_E}, \ \bgamma \in \reals^m.
    \end{split}
    \label{eq:ftj_gti}
\end{equation}
(Here we note that $h_j \in \reals^{\ell_H}$, $e_i \in \reals^{\ell_E}$, and $\beta_j, \gamma_i \in \reals$.) In words, the functions $f_{t+1}, g_t$ act component-wise on their inputs and depend only on the indices of the column block/row block. 

The matrices $\sB_{t+1} \in \reals^{\ell_E \times \ell_H}$ and $\sD_t \in \reals^{\ell_H \times \ell_E}$ in  \eqref{eq:Et_update} and \eqref{eq:Ht_update} are defined as:
\begin{align}
    &\sB_{t+1} = \frac{1}{m} \sum_{j=1}^n \tf_{t+1}'(h^{t+1}_j, \beta_j, \sfc),\nonumber\\
    &\sD_t = \frac{1}{m} \sum_{i=1}^m \tg_t'(e^t_i, \gamma_i, \sfr),
\end{align}
where $\tf_{t+1}'$ and $\tg_t'$ denote the Jacobians with respect to the first argument.

\emph{Assumptions}: 
\begin{enumerate}
\item As the dimensions $m, n \to \infty$, the aspect ratio $\frac{m}{n} \to \delta >0$. We emphasize that $\ell_E, \ell_H$ as well as $\Lr, \Lc$ are positive integers that do not scale with $n$ as $m, n \to \infty$. 

\item The functions $\tf_t( \cdot \, , \cdot \, , \sfc)$ and 
$\tg_t(\cdot \, , \cdot \, , \sfr)$ are Lipschitz for $t \ge 1$.

\item  For  $\sfc \in [\Lc]$ and $\sfr \in [\Lr]$, let $\nu(\bbeta_\sfc)$ and $\pi(\bgamma_\sfr)$ denote the  empirical distributions of $\bbeta_\sfc  \equiv  (\beta_j)_{j \in [\mc{J}_{\sfc}]}  \in \reals^{n/\C}$ and 
$\bgamma_\sfr \equiv (\gamma_i)_{i \in [\mc{I}_{\sfr}]} \in \reals^{m/\R}$, respectively. Then, for some $k \in [2, \infty)$, there exist scalar random variables $\bar{\beta}_{\sfc} \sim \nu_{\sfc}$   and $\bar{\gamma}_{\sfr} \sim \pi_{\sfr}$,
with $\E\{ | \bar{\beta}_{\sfc} |^k \} < \infty$ and $\E\{ | \bar{\gamma}_{\sfr} |^k \} <\infty$, such that $d_k(\nu(\bbeta_\sfc),\nu_{\sfc}) \rightarrow 0$ and $d_k(\pi(\bgamma_\sfr), \pi_{\sfr}) \rightarrow 0$ almost surely.  Here $d_k(P,Q)$ is the $k$-Wasserstein distance between distributions $P, Q$ defined on the same Euclidean probability space.

\item  For $\sfc \in [\C]$, we assume there exists a symmetric non-negative definite $\hat{\Sigma}^{0, \sfc} \in \reals^{\ell_E \times \ell_E}$ such that we almost surely have 
\begin{align}
\frac{1}{\delta  } \lim_{n \to \infty}   \frac{1}{(n / \C)}(\bi^0_{\sfc})^{\sT} \bi^0_{\sfc} = \hat{\Sigma}^{0, \sfc},
\label{eq:hSig0}
\end{align}
where $\bi^0_\sfc  \equiv  (i^0_j)_{j \in [\mc{J}_{\sfc}]}  \in \reals^{(n/\C)}$.
\end{enumerate}

 Theorem \ref{thm:SE_gen} below states that for each $t \ge 1$, the empirical distribution of the rows of $\be^t$ converge to the law of a Gaussian random vector $\sim \normal(0, \Sigma^t)$. Similarly, the empirical distribution of the rows of $\bh^t$ converge to the law of a Gaussian random vector $\sim \normal(0, \Omega^t)$. Here the covariance matrices $\Sigma^t \in \reals^{\ell_E \times \ell_E}$ and $\Omega^t \in \reals^{\ell_H \times \ell_H}$ are defined by the state evolution recursion, described below.

\paragraph{State Evolution.} 
For $t\ge 1$,  we recursively define the matrices $\Omega^{t} \in \reals^{\ell_H \times \ell_H}$ and $\Sigma^{t} \in \reals^{\ell_E \times \ell_E}$ as follows. Given $\Omega^t,\Sigma^t$, compute:
\begin{align}
    \Omega^{t+1} = \frac{1}{\R} \sum_{\sfr \in [\Lr]}  \hat{\Omega}^{t+1, \sfr} \, ,
    \label{eq:Omega_t1}
\end{align}
where for $\sfr \in 
[\Lr]$,
\begin{align}
    \hat{\Omega}^{t+1, \sfr} = \E\left\{ \tg_t(G^t_{\sfr}, \bar{\gamma}_\sfr,  \, \sfr)  \, \tg_t(G^t_{\sfr}, \bar{\gamma}_\sfr,  \, \sfr)^{\sT} \right\}, \quad 
    G^t_{\sfr} \sim \normal(0, \Sigma^t) \,  \perp \, \bar{\gamma}_\sfr \sim \pi_{\sfr}.
\end{align}
Next, 
\begin{align}
    \Sigma^{t+1} = \frac{1}{\C} \sum_{\sfc \in [\Lc]}   \hat{\Sigma}^{t+1, \sfc} \, ,
    \label{eq:Sigma_t1}
\end{align}
where for $\sfc \in  [\Lc]$,
\begin{align}
    \hat{\Sigma}^{t+1, \sfc} = \frac{1}{\delta}\E\left\{ \tf_{t+1}(\bar{G}^{t+1}_{\sfc}, \bar{\beta}_\sfc,  \, \sfc)  \, \tf_{t+1}(\bar{G}^{t+1}_{\sfc}, \bar{\beta}_\sfc,  \, \sfc)^{\sT} \right\}, \quad 
    \bar{G}^{t+1}_{\sfc} \sim \normal(0, \Omega^{t+1}) \,  \perp \, \bar{\beta}_\sfc \sim \nu_{\sfc}.
    \label{eq:Gt1_c}
\end{align}
The state evolution is initialized with $\Sigma^{0} = \frac{1}{\C} \sum_{\sfc \in [\Lc]}
\hat{\Sigma}^{0, \sfc}$, where $\hat{\Sigma}^{0, \sfc}$ is defined in \eqref{eq:hSig0}.

\begin{theorem}
Consider the abstract AMP recursion in \eqref{eq:Ht_update}-\eqref{eq:Et_update}, with the Assumptions (1)-(4) stated above. Let $\Psi: \reals^{\ell_H} \times \reals \to \reals$ and $\Phi: \reals^{\ell_E} \times \reals \to \reals$ be any pseudo-Lipschitz functions. Then, for all $t \ge 1$ and each $\sfr \in [\Lr]$ and $\sfc \in [\Lc]$, we almost surely have:
\begin{align}
 & \lim_{n \to \infty} \,  \frac{1}{n/\C} \sum_{j \in  \mc{J}_\sfc} \Psi(h^t_j \, , \beta_j)  = \E\{ \Psi( \bar{G}^t_{\sfc} \, , \bar{\beta}_\sfc) \}, 
\quad   \bar{G}^t_{\sfc} \sim \normal(0, \Omega^t) \,  \perp \, \bar{\beta}_\sfc \sim \nu_{\sfc} \, ,     \label{eq:Ht_result} \\
 &  \lim_{m \to \infty} \,  \frac{1}{m/\R} \sum_{i \in \mc{I}_\sfr } \Phi(e^t_i \, , \gamma_i)  = \E\{ \Phi( G^{t}_{\sfr} \, , \bar{\gamma}_\sfr) \}, \quad    G^t_{\sfr} \sim \normal(0, \Sigma^t) \,  \perp \, \bar{\gamma}_\sfr \sim \pi_{\sfr} \, .
      \label{eq:Et_result}
\end{align}
\label{thm:SE_gen}
\end{theorem}

Theorem \ref{thm:SE_gen} can be proved by applying \cite[Theorem 1]{javanmard2013state}, which gives an analogous state evolution result for an AMP recursion defined via a \emph{symmetric} Gaussian matrix. To obtain Theorem \ref{thm:SE_gen} for the AMP recursion \eqref{eq:Ht_update}-\eqref{eq:Et_update} defined via the non-symmetric  $\bM$, we  consider an AMP recursion defined on a symmetric matrix $\bM^{\text{sym}} \in \reals^{(m+n) \times (m+n)}$ defined as follows. Let $\bM_1 \in \reals^{m \times m}$ and $\bM_2 \in \reals^{n \times n}$ each be matrices with i.i.d.~$\normal(0, \frac{1}{2m})$ entries. Then
\begin{equation}
     \bM^{\text{sym}} =     \sqrt{\frac{\delta}{1+ \delta}} \begin{bmatrix}
& \bM_1 + \bM_1^{\sT} & \bM \\
& \bM^{\sT} & \bM_2 + \bM_2^{\sT}
    \end{bmatrix}
\end{equation}
is a symmetric $\text{GOE}(m+n)$ matrix. Using steps similar to those in \cite[Sec. 3.5]{javanmard2013state}, we can define an AMP recursion for $\bM^{\text{sym}}$ such that the $(2t+1)$th iterate gives $\be^t$ and the $(2t+2)$th iterate gives $\bh^{t+1}$. Using this together with the state evolution result for a symmetric AMP recursion in \cite[Theorem 1]{javanmard2013state}, we obtain the convergence results in \eqref{eq:Ht_result}, \eqref{eq:Et_result}. We omit the details as they are similar to  \cite[Sec. 3.5]{javanmard2013state}.
We mention that Theorem \ref{thm:SE_gen}  can also be proved using the graph-based AMP framework in \cite{Ger21}.

\begin{remark}
\label{rem:det_onsager}
Theorem \ref{thm:SE_gen} also holds for the modified AMP recursion where the matrices $\sB_{t+1}$ and $\sD_t$ in \eqref{eq:Et_update} and \eqref{eq:Ht_update} are replaced by their limiting values determined by state evolution. That is, for $t \ge 0$:

\begin{align}
    \sB_{t+1} \rightarrow \bar{\sB}_{t+1} &:= \frac{1}{\delta} \frac{1}{\C} \sum_{\sfc =1}^{\C} \E \{ \tf_{t+1}'(\bar{G}^{t+1}_{\sfc}, \bar{\beta}_\sfc,  \, \sfc) \} ,\\
    \sD_t \rightarrow \bar{\sD}_t &:= \frac{1}{\R} \sum_{\sfr =1}^{\R} \E \{ \tg_t'(G^t_{\sfr}, \, \bar{\gamma}_\sfr,  \, \sfr) \}.
    \label{eq:det_onsager}
\end{align}
\end{remark}

\subsubsection{Proof of Theorem \ref{thm:gen_sc_gamp_SE} via reduction of SC-GAMP to Abstract AMP}
\label{subsec:proof_main}
We define the i.i.d.~Gaussian matrix $\bM \in \reals^{m \times n}$ in terms of the spatially coupled matrix  $\bA$ as follows. For $i \in [m], \, j \in [n]$,
\begin{equation}
    M_{ij} =
    \begin{cases}
        \frac{A_{ij}}{\sqrt{\Lr \, W_{\sfr(i), \sfc(j)} }}, & \quad \text{ if } W_{\sfr(i), \sfc(j)} \neq 0, \\
        \stackrel{\text{indep.}}{\sim}\normal(0, \frac{1}{m}), & \quad \text{ otherwise }.
    \end{cases}
    \label{eq:M_def}
\end{equation}
From \eqref{eq:construct_sc_A}, we have that $M_{ij} \sim_{\normalfont\text{iid}} \normal(0, \frac{1}{m})$ for $i \in [m], \, j \in [n]$. Using the matrix $\bM$, we define an abstract AMP recursion for the form in \eqref{eq:Et_update}-\eqref{eq:ftj_gti}, with iterates 
$\be^t \in \reals^{m \times 2 \Lr}$ and $\bh^{t+1} \in \reals^{n \times \Lc}$. This is done via the following choice of functions $\tf_t : \reals^{\Lc} \times \reals \times [\Lc] \to \reals^{2 \Lr}$ and 
$\tg_t: \reals^{2 \Lr} \times \reals \times [\Lr] \to \reals^{\Lc}$,  for $t \ge 0$. In \eqref{eq:ftj_gti}, we  take $\bbeta := \bx $ and $\bgamma := \bveps$, recalling from \eqref{eq:GLM_def} that $\by  = \varphi(\bA\bx, \, \bveps)$, and let
\begin{align}
    & \tf_{t}(h_j, \, x_j, \ \sfc )  = \sqrt{\Lr} \Big[ \, \gin(h_{j, \sfc} + \mu_{\sfc}^q(t)  \, x_j, \sfc \, ; \, t ) \, [\sqrt{W_{1\sfc}}, \ldots, \sqrt{W}_{\Lr \sfc}]  , \   x_j[\sqrt{W_{1\sfc}}, \ldots, \sqrt{W_{\Lr \sfc}}] \Big],
    \nonumber \\ 
    & \hspace{30em} \text{ for } j \in \mc{J}_{\sfc}, \ h_j \in \reals^{\Lc}, \label{eq:ft_choice}  \\
    & \tg_{t}(e_i, \, \veps_i, \ \sfr )   = \sqrt{\Lr}
\,     \gout(e_{i, \sfr}, \varphi(e_{i, \sfr+ \Lr}, \, \varepsilon_i) , \sfr ; t) \, [\sqrt{W_{\sfr 1}} \alpha^q_{1}(t+1), \,  \sqrt{W_{\sfr 2}} \,  \alpha^q_{2}(t+1), \ldots \sqrt{W_{\sfr \Lc}} \, \alpha^q_{\Lc}(t+1)] \nonumber  \\
& \hspace{30em}  \text{ for } i \in \mc{I}_{\sfr}, \ e_i \in \reals^{2\Lr}.
\label{eq:gt_choice}
\end{align}
Here the notation $h_{j, \sfc}$ refers to the $\sfc$th component of the vector $h_j \in \reals^{\Lc}$. Similarly, $e_{i, \sfr}$ and $e_{i, \sfr+\Lr}$ are components of the vector $e_{i} \in \reals^{2\Lr}$.  We also recall that $\mu^q_\sfc(t)$ is the GAMP state evolution parameter defined according to \eqref{eq:muq_t1_def}. 
We note that the vectors in  \eqref{eq:ft_choice}-\eqref{eq:gt_choice} are treated as row vectors as they represent the $j$th row of the functions $f_t, g_t$, as given in \eqref{eq:ftj_gti}.

Consider the AMP algorithm \eqref{eq:Et_update}-\eqref{eq:Ht_update} defined via the matrix $\bM$ in \eqref{eq:M_def}, the functions in \eqref{eq:ft_choice}-\eqref{eq:gt_choice}, and the matrices $\sB_t, \sD_t$ replaced by their deterministic limiting values $\bar{\sB}_t, \bar{\sD}_t$, respectively (see \eqref{eq:det_onsager}). The algorithm is initialized with $\bi^0 \equiv  f_0(\bh^0, \, \bbeta) \in \reals^{n \times 2 \R}$ and $\be^0 = \bM \bi^0$, where the $j$th row
\begin{align}
    i_j^0 = \sqrt{\R} \Big[ \,  x_j(0) \, [\sqrt{W_{1\sfc}}, \ldots, \sqrt{W}_{\Lr \sfc}] \,  ,  \,    x_j \, [\sqrt{W_{1\sfc}}, \ldots, \sqrt{W_{\Lr \sfc}}] \Big] \quad \text{ for }j\in\mc{J}_\sfc, \  i_j^0\in\reals^{1\times 2\R}.
    \label{eq:i0_init}
\end{align}
With this choice of functions, the state evolution recursion is as follows. Given $\Sigma^t \in \reals^{2\R \times 2\R}$ and $\Omega^{t} \in \reals^{\C \times \C}$,  the matrices for iteration $(t+1)$ are computed according to \eqref{eq:Omega_t1}-\eqref{eq:Gt1_c}.  For  $\sfr \in [\R]$, let $G^t_{\sfr} \sim \normal(0,  \Sigma^t) \ \perp \ \bar{\veps} \sim P_{\bar{\veps}}$. Then the entries of $\hat{\Omega}^{t+1, \sfr} \in \reals^{\C \times \C}$ are:
\begin{align}
\left[ \hat{\Omega}^{t+1, \sfr} \right]_{\sfc \, \sfc'} =& \R \,  \E\{ \gout(G^t_{\sfr, \sfr}, \, \varphi(G^t_{\sfr, \sfr+\R}, \bar{\veps} ), \, \sfr \, ; t )^2 \}
\,  \sqrt{W_{\sfr \sfc} W_{\sfr \sfc'}} \alpha^q_{\sfc}(t+1) \alpha^q_{\sfc'}(t+1), \quad  \sfc, \, \sfc' \in [\C].
\label{eq:hOmegat1_def}
\end{align}
Here $G^t_{\sfr, \sfr}$ and $G^t_{\sfr, \sfr+\R}$ denote the $\sfr$th  and $(\sfr+\R)$th components, respectively, of $G^t_{\sfr} \in \reals^{2\R}$. 
For $\sfc \in [\C]$, let $\bar{G}^t_{\sfc} \sim \normal(0, \Omega^t) \ \perp \ X \sim P_{X}$. The entries of $\hat{\Sigma}^{t+1} \in \reals^{2\R \times 2\R}$ are:
\begin{align}
  \left[  \hat{\Sigma}^{t+1, \sfc} \right]_{k \ell} =
  \begin{cases}
 \frac{1}{\delta}\R \E \left\{ \gin( \bar{G}^t_{\sfc, \sfc} + \mu^q_{\sfc}(t)X, \sfc \, ; t)^2 \right\} \sqrt{W_{k \sfc} W_{\ell \sfc}} \, , & k, \ell \in [\R],  \\
 \frac{1}{\delta}\R \E \left\{ X \gin( \bar{G}^t_{\sfc, \sfc} + \mu^q_{\sfc}(t)X, \sfc \, ; t) \right\} \sqrt{W_{k \sfc} W_{(\ell - \R) \sfc}} \, , & k \in [\R], \ \R+1 \le \ell \le 2\R,   \\
 \frac{1}{\delta}\R \E \left\{ X \gin( \bar{G}^t_{\sfc, \sfc} + \mu^q_{\sfc}(t)X, \sfc \, ; t ) \right\} \sqrt{W_{(k-\R) \, \sfc} W_{\ell \sfc}} \, , &  \R+1 \le k \le 2\R, \ , \ell \in [\R],  \\
   \frac{1}{\delta}\R \E \{ X^2 \} \sqrt{W_{(k-\R) \, \sfc} W_{(\ell -\R) \,  \sfc}} \, , &  \R+1 \le k, \ell \le 2\R.
  \end{cases}
  \label{eq:hSigmat1_def}
\end{align}
Here $\bar{G}^t_{\sfc, \sfc}$ denotes the $\sfc$th component of $\bar{G}^t_{\sfc} \in \reals^{\Lc}$. 
The state evolution is initialized with $\hat{\Sigma}^{0, \sfc} = \frac{1}{\delta  } \lim_{n \to \infty}   \frac{1}{(n / \C)}(\bi^0_{\sfc})^{\sT} \bi^0_{\sfc}$, for $\sfc \in [\C]$, with $\bi^0 \in \reals^{n \times 2\R}$  defined in \eqref{eq:i0_init}. Using Assumption (A0) (see \eqref{eq:init_limits}), the entries are given by 
\begin{align}
  \left[  \hat{\Sigma}^{0, \sfc} \right]_{k \ell} =
  \begin{cases}
 \frac{1}{\delta}\R \Xi^{\sfc}_{11} \sqrt{W_{k \sfc} W_{\ell \sfc}} \, , & k, \ell \in [\R],  \\
\frac{1}{\delta}\R \Xi^{\sfc}_{12} \sqrt{W_{k \sfc} W_{(\ell - \R) \sfc}} \, , & k \in [\R], \ \R+1 \le \ell \le 2\R,   \\
\frac{1}{\delta}\R \Xi^{\sfc}_{12}
\sqrt{W_{(k-\R) \, \sfc} W_{\ell \sfc}} \, , &  \R+1 \le k \le 2\R, \  \ell \in [\R],  \\
   \frac{1}{\delta}\R \Xi^{\sfc}_{22}\sqrt{W_{(k-\R) \, \sfc} W_{(\ell -\R) \,  \sfc}} \, , &  \R+1 \le k, \ell \le 2\R.
  \end{cases}
  \label{eq:hSigma0}
\end{align}
We then have $\Sigma^0 = \frac{1}{\C} \sum_{\sfc \in [\C]} \hat{\Sigma}^{0, \sfc}$.

To prove Theorem \ref{thm:gen_sc_gamp_SE}, we first show that for $t \ge 0$ and $\sfr \in [\R]$, $\sfc \in [\C]$:
\begin{align}
 &   \begin{bmatrix}  &  \left[  \Sigma^{t} \right]_{\sfr \sfr} &  \left[  \Sigma^{t} \right]_{\sfr \,  (\sfr+\R)}  \\
  &  \left[  \Sigma^{t} \right]_{(\sfr+\R) \,  \sfr}  &   \left[  \Sigma^{t} \right]_{(\sfr+\R) \,  (\sfr+\R)}
  \end{bmatrix} \,  = \,  \Lambda_\sfr(t),  \quad   \text{  and }   \quad  \left[  \Omega^{t+1} \right]_{\sfc \sfc} \, = \, \tq_{\sfc}(t+1),
  \label{eq:SE_equiv}
\end{align}
where the quantities on the right of the equalities are the state evolution parameters of SC-GAMP defined in \eqref{eq:tq_t1_def}-\eqref{eq:Lambda22}.
This implies that $(G^t_{\sfr, \sfr}, \, G^t_{\sfr, \sfr+R}) \disteq  (P_{\sfr}(t), Z_{\sfr})$ and $\bar{G}^{t+1}_{\sfc, \sfc} \disteq G^q_{\sfc}(t+1)$.
We then show 
that for $t \ge 0$:
\begin{align}
& e^t_{i, \sfr} = p_i(t), \quad e^t_{i, \sfr + \Lr} = z_i,
\quad \text{ for } i \in \mc{I}_{\sfr}, \ \sfr \in [\Lr], \label{eq:eti_claim}\\ 
& h^{t+1}_{j, \sfc} + \mu_{\sfc}^q(t+1)  \, x_j   = q_j(t+1), \quad \text{ for } j \in \mc{J}_{\sfc} , \ \sfc \in [\Lc].
\label{eq:ht1j_claim}
\end{align}
Here $p_i(t), q_j(t+1)$ denote the $i$th and $j$th component, respectively, of the spatially coupled GAMP iterates $\bp(t)$ and $\bq(t)$. Theorem \ref{thm:gen_sc_gamp_SE} follows by using \eqref{eq:SE_equiv}-\eqref{eq:ht1j_claim} in Theorem  \ref{thm:SE_gen}, the state evolution result for the abstract AMP.

We now prove  \eqref{eq:SE_equiv} and then \eqref{eq:eti_claim}-\eqref{eq:ht1j_claim}.
\paragraph{Proof of \eqref{eq:SE_equiv}.} For $t = 0$, from \eqref{eq:hSigma0} we have
\begin{equation}
\begin{split}
      \left[  \hat{\Sigma}^{0, \sfc} \right]_{\sfr \sfr}
  &=  \frac{1}{\delta}\R \, \Xi^{\sfc}_{11} W_{\sfr \sfc} \, , \qquad 
   \left[  \hat{\Sigma}^{0, \sfc} \right]_{\sfr (\sfr+R)}
   = \frac{1}{\delta}\R \, \Xi^{\sfc}_{12} W_{\sfr \sfc},\\
   \left[  \hat{\Sigma}^{0, \sfc} \right]_{(\sfr+R) \sfr}
  &=  \frac{1}{\delta}\R \, \Xi^{\sfc}_{12} W_{\sfr \sfc} \, , \qquad 
   \left[  \hat{\Sigma}^{0, \sfc} \right]_{(\sfr+R) (\sfr+R)}
   = \frac{1}{\delta}\R \, \Xi^{\sfc}_{22} W_{\sfr \sfc}.
   \end{split}
\end{equation}
Therefore
\begin{equation}
\begin{split}
   &\left[  \Sigma^{0} \right]_{\sfr \sfr}
  =  \frac{1}{\delta} \frac{\R}{\C} \sum_{\sfc \in [\C]} \Xi^{\sfc}_{11} W_{\sfr \sfc} \, , \quad 
   \hspace{-0.5em}\left[ \Sigma^{0} \right]_{\sfr (\sfr+R)}=\left[ \Sigma^{0} \right]_{ (\sfr+R)\sfr}
   = \frac{1}{\delta} \frac{\R}{\C} \sum_{\sfc \in [\C]} \Xi^{\sfc}_{12} W_{\sfr \sfc} \, , \\
   &\left[ \Sigma^{0} \right]_{ (\sfr+R)(\sfr+R)}
   = \frac{1}{\delta} \frac{\R}{\C} \sum_{\sfc \in [\C]} \Xi^{\sfc}_{22} W_{\sfr \sfc}.
   \end{split}
   \label{eq:Sig0_step}
\end{equation}
Recalling the definition of $\Lambda_\sfr(0)$ from \eqref{eq:Lambda0_def} and noting that $\delin = \delta \C/\R$, we see that the first claim in \eqref{eq:SE_equiv} holds for $t=0$. Assume towards induction that for some $t \ge 0$ we have 
\begin{align}
 &   \begin{bmatrix}  &  \left[  \Sigma^{t} \right]_{\sfr \sfr} &  \left[  \Sigma^{t} \right]_{\sfr \,  (\sfr+\R)}  \\
  &  \left[  \Sigma^{t} \right]_{(\sfr+\R) \,  \sfr}  &   \left[  \Sigma^{t} \right]_{(\sfr+\R) \,  (\sfr+\R)}
  \end{bmatrix}  = \,  \Lambda_\sfr(t), 
 \quad  \text{ and  } \quad \left[  \Omega^{t} \right]_{\sfc \sfc} \, = \, \tq_{\sfc}(t).
 \label{eq:SE_equiv_hyp}
\end{align}
(The second part of the induction hypothesis \eqref{eq:SE_equiv_hyp} is ignored for $t=0$.) From \eqref{eq:hOmegat1_def} we have: 
\begin{align}
\left[ \Omega^{t+1} \right]_{\sfc \sfc} = 
\frac{1}{\R} \sum_{\sfr \in [\R]} \left[ \hat{\Omega}^{t+1, \sfr} \right]_{\sfc \, \sfc}
&  =  (\alpha^q_{\sfc}(t+1))^2 \sum_{\sfr \in [\R]}
W_{\sfr \sfc}  \E\{ \gout(G^t_{\sfr, \sfr}, \, \varphi(G^t_{\sfr, \sfr+\R}, \bar{\veps} ), \, \sfr \, ; t )^2 \} \nonumber \\
& = \tq_{\sfc}(t+1),
\label{eq:Omegat1_step}
\end{align}
where the last equality is obtained using the definitions of $\alpha^q_{\sfc}(t+1)$  in  \eqref{eq:alpha_qtq_def} and $\tq_{\sfc}(t)$ in  \eqref{eq:tq_t1_def}, and the induction hypothesis  which implies that  $(G^t_{\sfr, \sfr}, \, G^t_{\sfr, \sfr+R}) \disteq  (P_{\sfr}(t), Z_{\sfr})$.

Next, from \eqref{eq:hSigmat1_def} we have that 
\begin{align}
   & \left[  \Sigma^{t+1} \right]_{\sfr \sfr} \, 
= \frac{1}{\delta\C} \sum_{\sfc \in [\C]} \left[  \hat{\Sigma}^{t+1, \sfc} \right]_{\sfr \sfr} = \frac{\R}{\delta \C}  \sum_{\sfc \in [\C]} W_{\sfr \sfc} \, \E \left\{ \gin( \bar{G}^t_{\sfc, \sfc} + \mu^q_{\sfc}(t)X, \sfc \, ; t)^2 \right\} \,  , \label{eq:Sigt1_rr}    \\
   & \left[  \Sigma^{t+1} \right]_{\sfr \, (\sfr+\R)}
   =   \left[  \Sigma^{t+1} \right]_{(\sfr+\R) \, \sfr}  = \frac{1}{\delta\C} \sum_{\sfc \in [\C]} \left[  \hat{\Sigma}^{t+1, \sfc} \right]_{\sfr \, (\sfr+\R)} \nonumber \\
& \hspace{1.8in}= \frac{\R}{\delta \C}  \sum_{\sfc \in [\C]} W_{\sfr \sfc} \, \E \left\{ X \gin( \bar{G}^t_{\sfc, \sfc} + \mu^q_{\sfc}(t)X, \sfc \, ; t) \right\}, \label{eq:Sigt1_rrR}  \\
   & \left[  \Sigma^{t+1} \right]_{(\sfr+\R) \,  (\sfr+\R)} \, 
   =  \frac{1}{\delta \C} \sum_{\sfc \in [\C]} \left[  \hat{\Sigma}^{t+1, \sfc} \right]_{(\sfr + \R) \,  (\sfr+\R)} =\frac{\R \,  \E\{ X^2\}}{\delta \C}  \sum_{\sfc \in [\C]} W_{\sfr \sfc}  = \E\{ Z_{\sfr}^2\}. \label{eq:Sigt1_rRrR}
\end{align}
From the induction hypothesis, we have that $\bar{G}^t_{\sfc, \sfc} = G_{\sfc}^q(t)$. Therefore, comparing the expressions in \eqref{eq:Sigt1_rr}-\eqref{eq:Sigt1_rRrR} with the entries of $\Lambda_{\sfr}(t+1)$ defined in \eqref{eq:Lambda11}-\eqref{eq:Lambda22}, we conclude that
\begin{align}
  &   \begin{bmatrix}  &  \left[  \Sigma^{t+1} \right]_{\sfr \sfr} &  \left[  \Sigma^{t+1} \right]_{\sfr \,  (\sfr+\R)}  \\
  &  \left[  \Sigma^{t+1} \right]_{(\sfr+\R) \,  \sfr}  &   \left[  \Sigma^{t+1} \right]_{(\sfr+\R) \,  (\sfr+\R)}
  \end{bmatrix} \,  = \,  \Lambda_\sfr(t+1).\label{eq:Sigt1_full}
\end{align}
This completes the induction step and hence, proves \eqref{eq:SE_equiv}. 

\paragraph{Proof of \eqref{eq:eti_claim} and \eqref{eq:ht1j_claim}.}  At $t=0$,   the algorithm is initialized with $\bi^0$ defined in \eqref{eq:i0_init},  and   $\be^0 = \bM  \bi^0 \in  \reals^{m \times 2\R}$. For $\sfr \in [\R]$ and $i \in \mc{I}_\sfr$, consider the $\sfr$th and the $(\sfr+R)$th entries of the  $i$th row of $\be^0$. Writing the sum over $j$ from $1$ to $n$ as a double sum, we have that:
\begin{align}
& e^0_{i, \sfr}  = \sum_{\sfc = 1}^{\C}  \sum_{j \in \mc{J}_\sfc}  M_{ij} \sqrt{R \, W_{\sfr\sfc}} \, x_j(0),  \nonumber \\
& e^0_{i, \sfr + \Lr}  = \sum_{\sfc = 1}^{\C}  \sum_{j \in \mc{J}_\sfc} M_{ij} \sqrt{R \, W_{\sfr\sfc}} \, x_j.  \label{eq:e0_irR} 
\end{align}
Recalling the definition of $M_{ij}$ in \eqref{eq:M_def},  we  have that  
\begin{align}
e^0_{i, \sfr} =  \sum_{j=1}^n A_{ij}x_j(0)  = p_i(0), \qquad e^0_{i, \sfr + \Lr} = \sum_{j=1}^n A_{ij}x_j  = z_i,
\end{align}
where we used $\bp(0) = \bA \bx(0)$ and $\bz = \bA \bx$. 
Assume towards induction that for some $t \ge 0$, we have:
\begin{align}
& e^t_{i, \sfr} = p_i(t), \quad e^t_{i, \sfr + \Lr} = z_i,
\quad \text{ for } i \in \mc{I}_{\sfr}, \ \sfr \in [\Lr], \label{eq:eti_hyp}\\ 
& h^{t}_{j, \sfc} + \mu_{\sfc}^q(t)  \, x_j   = q_j(t), \quad \text{ for } j \in \mc{J}_{\sfc} , \ \sfc \in [\Lc].
\label{eq:htj_hyp}
\end{align}
(For $t=0$, we only assume \eqref{eq:eti_hyp}.)

From \eqref{eq:Ht_update}, the $j$th row of $\bh^{t+1}$, for $j \in \mc{J}_{\sfc}$, is given by 
\begin{equation}
h^{t+1}_j = \sum_{i=1}^n M_{ij} \,  \tg_t(e_i^t,\veps_i, \sfr(i) )
\, - \, \tf_t(h_j^t, x_j, \sfc) \bar{\sD}_t^{\sT},
\label{eq:htj_exp}
\end{equation}
where from \eqref{eq:det_onsager},
\begin{align}
   \bar{\sD}_t =  \frac{1}{\R} \sum_{\sfr =1}^{\R} \E \{ \tg_t'(G^t_{\sfr}, \, \bar{\veps},  \, \sfr) \} \, \in \reals^{\Lc \times 2 \Lr }.
\end{align}
Here $G^t_{\sfr} \sim \normal(0, \Sigma^t)$ is independent of the noise distribution $W$, with $\Sigma^t \in \reals^{2\Lr \times 2\Lr}$ defined according to the state evolution recursion in 
\eqref{eq:Omega_t1}-\eqref{eq:Gt1_c}. Using the definition of $\tg_t$ in \eqref{eq:gt_choice}, the entries of its Jacobian are:  
\begin{align}
[ \tg_t'(G^t_{\sfr}, \, \bar{\veps},  \, \sfr) ]_{k \ell} 
= \begin{cases}
    \sqrt{R} \sqrt{W_{\sfr k}} \, \alpha^q_k(t+1) \, \goutpr(G^t_{\sfr, \sfr}, \, \varphi(G^t_{\sfr, \sfr+\R}, W ), \, \sfr \, ; t ), & \text{ for } k \in [\C],  \, \ell = \sfr, \\
    \sqrt{R} \sqrt{W_{\sfr k}} \,  \alpha^q_k(t+1) \, \partial_z \gout(G^t_{\sfr, \sfr}, \, \varphi(G^t_{\sfr, \sfr+\R}, \bar{\veps}),  \, \sfr \, ; t ), & \text{ for } k \in [\C],  \, \ell = \sfr + \R, \\
     0, & \text{ otherwise.}
\end{cases}
\end{align}
Here we recall that $ \goutpr(p, \varphi(z, \veps) \, ; \, t)$ denotes the derivative with respect to  the first argument and $\partial_z \, \gout(p, \varphi(z, \veps) \, ; \, t)$ the derivative with respect to $z$.
Using this, we obtain that 
\begin{align}
   &[ \bar{\sD}_t ]_{\sfc \sfr}
   =  
    \frac{1}{\sqrt{\R}} 
        \sqrt{W_{\sfr \sfc}} \, \alpha^q_{\sfc}(t+1) \,
        \E\{ \goutpr(G^t_{\sfr, \sfr}, \, \varphi(G^t_{\sfr, \sfr+\R}, \bar{\veps}), \, \sfr \, ; t ) \},\quad  \sfc \in [\C],   \sfr \in [\R], \nonumber \\
        &[ \bar{\sD}_t ]_{\sfc (\sfr+\R)} =  \frac{1}{\sqrt{\R}} 
        \sqrt{W_{\sfr \sfc}} \, \alpha^q_{\sfc}(t+1) \, \E\{ \partial_z \gout(G^t_{\sfr, \sfr}, \, \varphi(G^t_{\sfr, \sfr+\R}, \bar{\veps}), \, \sfr \, ; t ) \}, \quad \sfc \in [\C],   \sfr \in [\R]. \label{eq:barDt_exp01}
\end{align}
Using \eqref{eq:barDt_exp01} in \eqref{eq:htj_exp} along with the definition of $\tf_t$ from \eqref{eq:ft_choice}, we obtain
\begin{equation}
\begin{split}
   &h^{t+1}_{j, \sfc}  =  \sum_{ \sfr = 1}^{\R} \sum_{i \in \mc{I}_\sfr} M_{ij}  \sqrt{\R W_{\sfr c}} \, \alpha_{\sfc}^q(t+1) \, \gout(e^t_{i, \sfr}, \varphi(e^t_{i, \sfr + \Lr}, \, \veps_i) , \sfr ; t)  \\
   &   \qquad  - \, \gin(h_{j, \sfc}^t + \mu_{\sfc}^q(t)  \, x_j, \sfc \, ; \, t ) \,  \alpha_{\sfc}^q(t+1) \sum_{\sfr=1}^{\Lr} W_{\sfr \sfc} \E\{ \goutpr(G^t_{\sfr, \sfr}, \, \varphi(G^t_{\sfr, \sfr+\R}, \bar{\veps}), \, \sfr \, ; t ) \} 
   \,  \\
   & \qquad   - \, x_j 
    \alpha_{\sfc}^q(t+1) \sum_{\sfr=1}^{\Lr} W_{\sfr \sfc} \E\{ \partial_z \gout(G^t_{\sfr, \sfr}, \, \varphi(G^t_{\sfr, \sfr+\R}, \bar{\veps}), \, \sfr \, ; t ) \}.
    \end{split}
    \label{eq:ht1_jc0}
\end{equation}
By the induction hypothesis, we have that
$e^{t}_{i, \sfr} = p_i(t)$ and $ e^t_{i, \sfr + \Lr} =  z_i$, for $i \in \mc{I}_{\sfr}$, and $h^t_{j, \sfc} + \mu^q_{\sfc}(t)x_j    = q_j(t) $ for $j \in \mc{J}_{\sfc}$. (For $t=0$,  we only assume the hypothesis on $e^t$, and the formula in \eqref{eq:ht1_jc0} holds for $h^{1}_{j, \sfc}$ with  the term $\gin(h_{j, \sfc}^t + \mu_{\sfc}^q(t)  x_j, \sfc  ; t)$  replaced by $x_j(0)$.) Moreover by \eqref{eq:SE_equiv}, we have $(G^t_{\sfr, \sfr+\Lr}, G^t_{\sfr, \sfr}) \disteq (Z_\sfr, P_\sfr(t))$.  Using these along with the formulas for $\mu^q_{\sfc}(t+1)$ and $\alpha_{\sfc}^q(t+1)$ in \eqref{eq:muq_t1_def} and \eqref{eq:alpha_qtq_def}, we obtain:
\begin{align}
    h^{t+1}_{j, \sfc} & = \alpha_{\sfc}^q(t+1) \,  \sum_{i=1}^n A_{ij} \gout(p_i(t), \, y_i, \, \sfr(i); \, t) \,  + \, \gin(q_j(t), \sfc; \, t)
    \, - \, \mu^q_{\sfc}(t+1)  x_j \nonumber \\
    & = q_j(t+1) \, - \, \mu^q_{\sfc}(t+1)  x_j.
    \label{eq:ht1_step}
\end{align}
(For $t=0$, the term $\gin(q_j(t), \sfc; \, t)$ is replaced by $x_j(0)$, and both equalities in  \eqref{eq:ht1_step} still hold because of the SC-GAMP initialization $\bx(0) \equiv \bargin(\bq(0) \, ; \, 0 )$.)

Next consider the $i$th row of $\be^{t+1}$, for $i \in \mc{I}_\sfr$, which from \eqref{eq:Et_update}, is given by 
\begin{equation}
e^{t+1}_i = \sum_{j=1}^n M_{ij} \, \tf_{t+1}(h_j^{t+1}, x_j, \sfc(j)) \, - \, \tg_{t}(e^{t}_i, \gamma_i, \sfr) \,  \bar{\sB}_{t+1}^{\sT}
\label{eq:eti_exp}
\end{equation}
where from \eqref{eq:det_onsager},  
\begin{align}
   \bar{\sB}_{t+1} = \frac{1}{\delta} \frac{1}{\C} \sum_{\sfc =1}^{\C} \E \{ \tf_{t+1}'(\bar{G}^{t+1}_{\sfc}, X,  \, \sfc) \} \, \in \reals^{2 \Lr \times \Lc}.
\end{align}
Here $\bar{G}^{t+1}_{\sfc} \sim \normal(0, \Omega^{t+1})$ is independent of the signal prior $X$, with $\Omega^{t+1} \in \reals^{\Lc \times \Lc}$ defined according to the state evolution recursion in 
\eqref{eq:Omega_t1}-\eqref{eq:Gt1_c}. From the definition of $\tf_t$ in \eqref{eq:ft_choice}, the entries of its Jacobian are: 
\begin{equation}
   [\tf_{t+1}'(\bar{G}^{t+1}_{\sfc}, X,  \, \sfc)]_{k   \ell} = \begin{cases}
        \sqrt{\Lr} \, \ginpr(\bar{G}^{t+1}_{\sfc, \sfc} + \mu^q_{\sfc}(t+1) X, \sfc \, ; t+1) \sqrt{W_{k \sfc}},  &    \text{ for } k \in [\Lr],  \ell= \sfc, \\
       0, &  \text{ otherwise}.
   \end{cases}
\end{equation}
Therefore, 
\begin{align}
   [ \bar{\sB}_{t+1} ]_{k \sfc}  =  
   \frac{1}{\delta} \frac{1}{\C} 
   \begin{cases}
        \sqrt{\Lr} \, \E\{ \ginpr(\bar{G}^{t+1}_{\sfc, \sfc} + \mu^q_{\sfc}(t+1)X, \, \sfc \, ; t+1) \} \sqrt{W_{k \sfc}},  &    \text{ for } k \in [\Lr],  \sfc \in [\Lc], \\
       0, &  \text{ otherwise}.
   \end{cases}
   \label{eq:barBt_exp}
\end{align}
We now use \eqref{eq:barBt_exp} in \eqref{eq:eti_exp} along with the definition of $\tg_{t}$ in \eqref{eq:gt_choice}. For $i \in \mc{I}_{\sfr}$, $r \in [\Lr]$, we have:
\begin{align}
& e^{t+1}_{i, \sfr + \Lr}  = \sum_{\sfc = 1}^{\C}  \sum_{j \in \mc{J}_\sfc} M_{ij} \sqrt{R \, W_{\sfr\sfc}} \, x_j = \sum_{j=1}^n A_{ij}x_j  = z_i \, , \label{eq:et_irR} \\
& e^{t+1}_{i, \sfr} = \sum_{\sfc = 1}^{\C}  \sum_{j \in \mc{J}_\sfc}  M_{ij} \sqrt{R \, W_{\sfr\sfc}}
\, \gin( h^{t+1}_{j, \sfc} + \mu^q_{\sfc}(t+1)x_j, \sfc; \, t+1) \nonumber \\
& \qquad  - \gout(e^{t}_{i, \sfr}, \varphi(e^t_{i, \sfr+ \Lr}, \, \veps_i) , \sfr \, ;  t) \,  \frac{\Lr}{\delta \Lc} \sum_{\sfc =1}^{\Lc} W_{\sfr \sfc} \,  \alpha_{\sfc}^q(t+1) \, \E\{ \ginpr(\bar{G}^{t+1}_{\sfc, \sfc} + \mu^q_{\sfc}(t+1)X, \sfc \, ; t+1) \}. \label{eq:et_ir}
\end{align}
Next, by the induction hypothesis we have that 
$e^{t}_{i, \sfr} = p_i(t)$ for $i \in \mc{I}_\sfr$, and we have shown above  that $h^{t+1}_{j, \sfc} + \mu^q_{\sfc}(t+1)x_j    = q_j(t+1) $  for $j \in \mc{J}_{\sfc}$. and from \eqref{eq:SE_equiv} we have $\bar{G}^{t+1}_{\sfc, \sfc} + \mu^q_{\sfc}(t+1)X \disteq Q_\sfc(t+1)$. Using these in \eqref{eq:et_ir} along with $\delin = \frac{\Lr}{\delta \Lc}$ and the formula for $\alpha_{\sfr}^p(t+1)$ from \eqref{eq:alpha_pt_def},  we obtain: 
\begin{align}
   e^{t+1}_{i, \sfr} &  =  \sum_{j=1}^n A_{ij} \, \gin( q_j(t+1), \sfc(j); \, t+1)
   \, - \, \gout(p_i(t), y_i,  \sfr \, ;  t) \,  \alpha^p_{\sfr}(t)\nonumber \\
   & = p_i(t+1),
\end{align}
where the last equality follows from \eqref{eq:SC-GAMP2}.

This completes the proof of the claims in \eqref{eq:eti_claim} and \eqref{eq:ht1j_claim}, and Theorem \ref{thm:gen_sc_gamp_SE} follows. \qed

\subsection{Proof of Theorem \ref{thm:SC_FPs}} \label{subsec:proof_SC_FPs}

\subsubsection{Proof of Part 1} 
From \eqref{eq:SC-GAMP-MSE}, for $t \ge 1$ we almost surely have:
\begin{equation}
\begin{split}
   &  \lim_{n \to \infty} \,  \frac{\| \bx - \bargin^*(\bq(t))\|^2}{n}  =  \frac{1}{\C} \sum_{\sfc=1}^{\C} \, \E\{ (X - \ginbayes(Q_{\sfc}(t), \sfc ; t ) )^2\} \ \text{ a.s.}
    \label{eq:Bayes-SC-GAMP-MSE}
    \end{split}
\end{equation}
Here we recall from \eqref{eq:PtQt_def} that   $Q_{\sfc}(t)  =  \mu^{q}_\sfc(t) X \, +  \, G^q_{\sfc}(t)$ where $X \sim P_X \, \perp \,   G^q_{\sfc}(t)\sim  \normal(0, \tau^q_{\sfc}(t))$. The following corollary  specifies the state evolution recursion for Bayes SC-GAMP. 
\begin{corr} 
\label{cor:BayesSE}
For  a base matrix $\bW$ satisfying \eqref{eq:Wrc_assumptions}, the state evolution parameters of the Bayes SC-GAMP (defined via the functions $\ginbayes,\goutbayes$ in \eqref{eq:ginbayes}-\eqref{eq:goutbayes}) satisfy the following for $t \ge 1$ and $\sfr \in [\R]$ and $\sfc \in [\C]$:
\begin{align}
   &  \mu_{\sfc}^q(t) = \mu_{\sfr}^p(t) =1, \label{eq:muc_mur_Bayes}\\
   &\frac{1}{ \tq_{\sfc}(t)} =  \sum_{\sfr =1}^{\R} W_{\sfr \sfc}  \frac{1}{\tp_{\sfr}(t-1)}\left(1  -     \frac{1}{\tp_{\sfr}(t-1)}\E_{P_{\sfr}, Y_{\sfr}}  \left\{\var{(Z_{\sfr} \mid P_{\sfr}, Y_\sfr ;  \, \tp_{\sfr}(t-1))} \right\}\right),
  \label{eq:mu_tau_q_Bayes} \\
    &  \tp_{\sfr}(t) = \frac{1}{\delin} \sum_{\sfc =1}^
    {\C}  W_{\sfr \sfc} \,  \mmse\big(\tq_{\sfc}(t)^{-1} \big),  \label{eq:mu_tau_p_Bayes} \\
       & \alpha^q_{\sfc}(t) = \tq_{\sfc}(t),   \qquad  \alpha^p_{\sfr}(t)   
       =  \tp_\sfr(t).   
       \label{eq:alpha_pq_bayes}
\end{align}
Here, for $s >0$,
\begin{align}
        \mmse\big(s \big) =& \E\left\{ \left(X \, - \, \E\Big\{  X  \mid \sqrt{s} X + \,  G \Big\} \right)^2 \right\}, \quad G \sim \normal(0,1) \, \perp \,  X \sim P_X,
        \label{eq:mmse_scalar_def}
\end{align}
and 
$Y_{\sfr}=\varphi(Z_{\sfr}, \, \bar{\veps})$ with $\veps \sim P_{\bar{\veps}} \perp Z_{\sfr}$. The term $\var{(Z_{\sfr} \mid P_{\sfr}, Y_\sfr ;  \tp_{\sfr}(t-1))}$ denotes the conditional variance computed with   $(P_{\sfr}, Z_{\sfr}) \sim \normal(0, \Lambda_r(t-1))$ where
\begin{align}
   & \Lambda_r(t-1) =   \begin{bmatrix}
   \E\{Z_{\sfr}^2 \} - \tp_{\sfr}(t-1)  &   \E\{Z_{\sfr}^2 \} - \tp_{\sfr}(t-1)  \\
  \E\{Z_{\sfr}^2 \} - \tp_{\sfr}(t-1)  &  \E\{Z_{\sfr}^2 \} 
   \end{bmatrix}.
   \label{eq:Lrt1_def}
\end{align}
The iteration \eqref{eq:mu_tau_q_Bayes}-\eqref{eq:mu_tau_p_Bayes} is initialized with 
\begin{align} 
\tp_{\sfr}(0) =  \frac{\var(X)}{\delin} \sum_{\sfc=1}^{\C} W_{\sfr \sfc} \,. 
\label{eq:BayesSEinit}
\end{align}
\end{corr}

The proof of the corollary is given in Appendix \ref{app:corr_proof}.  Using  Corollary \ref{cor:BayesSE} and recalling the expression for $\ginbayes$ from \eqref{eq:ginbayes}, the limiting MSE of Bayes SC-GAMP given by  \eqref{eq:Bayes-SC-GAMP-MSE} for $t \ge 1$ is given by:
\begin{align}
\lim_{n \to \infty} \,  \frac{\| \bx - \bargin^*(\bq(t))\|^2}{n} = \frac{1}{\C} \sum_{\sfc=1}^{\C}  \mmse\big(\tq_{\sfc}(t)^{-1} \big),
\label{eq:MSE_exp}
\end{align}
where $\tq_{\sfc}(t)$ for $c \in  [\C]$ is computed as specified in Corollary \ref{cor:BayesSE}. To prove \eqref{eq:omLam_MSE}, we  show that for all sufficiently large $t$, the term $\frac{1}{\C} \sum_{\sfc=1}^{\C}  \mmse\big(\tq_{\sfc}(t)^{-1} \big)$ is bounded by the right side of \eqref{eq:omLam_MSE}. We do this by analyzing the fixed point of the state evolution recursion in Corollary \ref{cor:BayesSE} using a general result by Yedla et al. \cite{yedla2014asimple} on fixed points of coupled recursions. 

\paragraph{Fixed points of a general coupled recursion.}

We review the coupled recursion analyzed in \cite{yedla2014asimple}. Consider a matrix $\boldsymbol{B}$ with dimensions $L_2\times L_1$ and $L_1 = L_2 + \omega - 1$, whose entries are:
\begin{equation}
    B_{j,k}\triangleq \begin{cases}
    \frac{1}{\omega}, & \text{if $j\leq k \leq j +\om- 1$}.\\
    0, & \text{otherwise.}\end{cases}
    \label{eq:coupledmatrix}
\end{equation} 
Let $f:\mathcal{Y}\to \mathcal{X}$ be a non-decreasing $C^1$ function  and $g:\mathcal{X}\to \mathcal{Y}$ be a strictly increasing $C^2$ function  with $y_\text{max}=g(x_{\max})$.  (For $k \ge 1$, a function is said to be $C^k$ if all derivatives up to order $k$ exist and are continuous.)
Define a coupled recursion of the form 
\begin{align}
    y_i{(t+1)}= &\  g\left(x_i{(t)}\right), \label{eq:coupled-rec1}\\
    x_i{(t+1)}= & \sum_{j=1}^{L_2} B_{j,i}\ f\left(\sum_{k=1}^{L_1} B_{j,k}y_k{(t+1)}\right), \label{eq:coupled-rec2}
\end{align} 
for $i \in [L_1]$.

The recursion is initialized by choosing $x_i{(0)} = x_{\max}$ for $i \in [L_1]$. This initialization, along with the monotonicity of $f$ and $g$, ensures that the coupled recursion converges to a fixed point \cite{yedla2014asimple}. The fixed point $\{ \lim_{t \to \infty} x_i(t) \}_{i \in [L_1]}$ is  characterized in terms of a \emph{potential} function $V(x)$ defined as:
\begin{equation}
    V(x) := xg(x) - G(x) - F(g(x)),
    \label{eq:PotFn_formula}
\end{equation}
where $F(x)=\int_{0}^x f(z) dz$ and $G(x)=\int_{0}^x g(z) dz$.
Then the fixed point of the coupled recursion can be bounded as follows.
\begin{lemma} \label{lem:Yedla1}\cite[Theorem 1]{yedla2014asimple}
For any $\gamma>0$, there is $\omega_0<\infty$ such that for all $\omega>\omega_0$ and $L_2\in [1,\infty]$, the fixed point $\{ \lim_{t \to \infty} x_i(t) \}_{i \in [L_1]}$ of the coupled recursion in \eqref{eq:coupled-rec1}-\eqref{eq:coupled-rec2} satisfies the upper bound
\begin{equation} \label{eq:YedlaThm1}
    \max_{i\in[1:L_1]} x_i{(\infty)}   \leq \max \left( \operatorname*{argmin}_{x\in \mathcal{X}} V(x)\right) +  \gamma.
\end{equation} 
\end{lemma}

\emph{Proof of  \eqref{eq:omLam_MSE}}: 
For the rest of the proof, we will take $\bW$ to be an $(\omega, \Lambda)$ base matrix, for which $\R = \Lambda+ \omega -1$ and $\C = \Lambda$.
We rewrite the state evolution recursion in \eqref{eq:mu_tau_q_Bayes}-\eqref{eq:mu_tau_p_Bayes}  in terms of parameters $\bar{x}_\sfr(t), \bar{y}_\sfr(t)$,  defined as follows for $\sfr \in [\R]$.
For $t \ge 0$, let:
\begin{align}
& \bar{x}_\sfr(t) := \delin \tp_{\sfr}(t), \label{eq:GAMP-cr0} \\
    & \bar{y}_\sfr(t+1) : =  \frac{- \delin}{\bar{x}_\sfr(t)}\Bigg(1    -  \frac{\delin}{\bar{x}_\sfr(t)}\E_{P_{\sfr}, Y_{\sfr}}  \left\{\var\left(Z_{\sfr} \mid P_{\sfr}, Y_{\sfr}; \frac{\bar{x}_\sfr(t)}{\delin}\right) \right\}\Bigg). \label{eq:GAMP-cr1}
     \end{align}
Then, from \eqref{eq:mu_tau_q_Bayes}-\eqref{eq:mu_tau_p_Bayes}, we have that:
\begin{align}
    \bar{x}_\sfr(t+1) = \sum_{\sfc =1}^{\C} W_{\sfr \sfc} 
    \mmse\Big(-\sum_{\sfr=1}^{\R} W_{\sfr \sfc} \bar{y}
_\sfr(t+1) \Big).
    \label{eq:GAMP-cr2}
\end{align}
In \eqref{eq:GAMP-cr1}, the joint distribution of $(Z_\sfr, P_{\sfr}, Y_{\sfr})$ is given by $Y_{\sfr}=\varphi(Z_{\sfr}, \, \bar{\veps})$ with $\veps \sim P_{\bar{\veps}} \perp Z_{\sfr}$, and  $(P_{\sfr}, Z_{\sfr})$ jointly Gaussian with mean zero and covariance matrix
\begin{align}
   & \begin{bmatrix}
   \E\{Z_{\sfr}^2 \} - \frac{\bar{x}_\sfr(t)}{\delin}  &   \E\{Z_{\sfr}^2 \} - \frac{\bar{x}_\sfr(t)}{\delin}  \\
  \E\{Z_{\sfr}^2 \} - \frac{\bar{x}_\sfr(t)}{\delin}  &  \E\{Z_{\sfr}^2 \} 
   \end{bmatrix}.
   \label{eq:Lrt1_def_x}
\end{align}
Here we recall from \eqref{eq:Lambda22} that
$\E\{Z_{\sfr}^2 \}  = \frac{\E\{ X^2 \}}{\delin} \sum_{\sfc =1}^{\C} W_{\sfr\sfc}$.  For an $(\omega, \Lambda)$ base matrix, from  Definition \ref{def:ome_lamb_rho}, we have
\begin{equation}
\E\{Z_{\sfr}^2 \} = 
\begin{cases}
     \frac{\E\{ X^2 \}}{\delin}  \frac{\sfr}{\omega} \, , &  1 \le \sfr \le (\omega-1), \\
    \frac{\E\{ X^2 \}}{\delin} \, , & \omega \le \sfr \le \Lambda,  \\
    \frac{\E\{ X^2 \}}{\delin}  \frac{\Lambda + \omega - \sfr}{\omega} \, , & \Lambda+1 \le \sfr \le \Lambda + \omega -1.
\end{cases}
\label{eq:EZr2_cases}
\end{equation}

We now define a slightly modified recursion with parameters $x_\sfr(t), y_\sfr(t)$, where the $\E\{ Z_{\sfr}^2 \}$ is replaced by $\E\{Z^2 \} = \E\{ X^2 \}/\delin$. For $\sfr \in [\R]$, let: 
\begin{align}
    & y_\sfr(t+1) : =  \frac{- \delin}{x_\sfr(t)}\Bigg(1  -  \frac{\delin}{x_\sfr(t)}\E_{P, Y}  \left\{\var\left(Z \mid P, Y; \frac{x_\sfr(t)}{\delin}\right) \right\}\Bigg). \label{eq:GAMP-cr1-mod} \\
     & x_\sfr(t+1) : = \sum_{\sfc =1}^{\C} W_{\sfr \sfc} 
    \mmse\Big(-\sum_{\sfr=1}^{\R} W_{\sfr \sfc} 
    y_\sfr(t+1) \Big).     \label{eq:GAMP-cr2-mod}
\end{align}
This recursion is initialized with $ x_{\sfr}(0) = \var(X)$ for $\sfr \in [\R]$.

 The only difference between the recursions \eqref{eq:GAMP-cr1}-\eqref{eq:GAMP-cr2} and \eqref{eq:GAMP-cr1-mod}-\eqref{eq:GAMP-cr2-mod} is that the joint distribution of $(Z,P,Y)$ in the latter is given by $Y=\varphi(Z, \, \bar{\veps})$ with $\bar{\veps} \sim P_{\bar{\veps}} \perp Z$, and  $(P, Z)$ jointly Gaussian with mean zero and covariance matrix
\begin{align}
   & \begin{bmatrix}
   \frac{\E\{ X^2 \}}{\delin} - \frac{x_\sfr(t)}{\delin}  &    \frac{\E\{ X^2 \}}{\delin} - \frac{x_\sfr(t)}{\delin}  \\
   \frac{\E\{ X^2 \}}{\delin} - \frac{x_\sfr(t)}{\delin}  &   \frac{\E\{ X^2 \}}{\delin}
   \end{bmatrix}.
   \label{eq:Lt1_def_x}
\end{align}
Comparing the covariance of $(P,Z)$ with that of $(P_{\sfr}, Z_{\sfr})$  (given in \eqref{eq:Lrt1_def_x}-\eqref{eq:EZr2_cases}), the key difference is that $\E\{Z_{\sfr}^2 \}$ is less than  $\E\{ X^2 \}/\delin$ for $1 \le \sfr \le (\omega-1)$ and for $\Lambda+1 \le \sfr \le \Lambda+\omega-1$, and equal to $\E\{ X^2 \}/\delin$ for all other $\sfr$.

The modified recursion \eqref{eq:GAMP-cr1-mod}-\eqref{eq:GAMP-cr2-mod} is an instance of the coupled recursion in 
\eqref{eq:coupledmatrix}-\eqref{eq:coupled-rec2}, which can be seen by taking $\boldsymbol{B} = \bW^{\sT}$, $L_2= \C = \Lambda$, and 
\begin{align}
   &  g(x) = \frac{-\delin}{x}\Bigg( 1 -  \frac{\delin}{x}\E_{P, Y}  \left\{\var\left(Z \mid P, Y; \frac{x}{\delin}\right) \right\} \Bigg), \label{eq:gdef} \\
   &  f(y)= \mmse(-y). \label{eq:fdef}
\end{align}

Letting $x_{\max} = \var(X)$ and $\mc{X} = [0, x_{\max}], \,  \mc{Y} = (-\infty, g(x_{\max})]$, we claim  that  $f:\mc{Y} \to \mc{X}$ is non-decreasing, and $g:\mc{X} \to \mc{Y}$ is strictly increasing. For $f$, this follows by noting that $y <0$ and the function $\mmse(\cdot)$ function  in \eqref{eq:mmse_scalar_def} is non-increasing in its argument \cite{guo2011estimation, barbier2016proof}. To show the monotonicity of $g(x)$ in $[0, \var(X)]$, we express it in an alternative form:
\begin{align}
    g(x) =  \frac{-\delin}{x} \,  \E\left( 
    \E\left\{  G \mid P, Y; \, \frac{x}{\delin} 
    \right \}^2\right),
    \label{eq:gx_alt}
\end{align}
where $G \sim  \normal(0, 1)$ $\perp$ $P \sim \normal(0, (\E X^2 - x)/\delin)$, $Z=P + \sqrt{\frac{x}{\delin}}G$ and $Y= \varphi(Z, \bar{\veps})$. The equivalence between the expressions in \eqref{eq:gdef} and \eqref{eq:gx_alt} is shown in \eqref{eq:alt_gout_exp}-\eqref{eq:alt_gout_exp1} in Appendix \ref{app:pot_fun_equiv}.

The function $g(x)$ and its derivative were analyzed in \cite[Appendix B.2]{barbier2019optimal}.  (The function $\Psi'_{P_{\text{out}}}(q)$ in Proposition 20 of that paper equals $-\frac{1}{2}g(x)$, with $q=\E X^2  - \frac{x}{\delin}$.)
In particular, \cite[Proposition 18, Eq. (199)]{barbier2019optimal} shows that $g'(x) >0$ for $x \in [0, \var(X)]$, hence $g(x)$ is increasing.

Since the recursion $\{x_{\sfr}(t), y_{\sfr}(t)\}_{t \ge 1}$ is a special case of the coupled recursion, with  the conditions satisifed, it converges to a fixed point which can be characterized using Lemma \ref{lem:Yedla1}. Moreover, using the I-MMSE relationship \cite{guo2005mutual}
\[ 
\frac{d}{ds} I(X \, ; \, \sqrt{s} X + G)  = \frac{1}{2} \mmse(s), \quad G \sim \normal(0,1) \, \perp \,  X \sim P_X,
\]
the potential function $V(x)$ computed according to \eqref{eq:PotFn_formula} equals $U(x; \delin)$ defined in \eqref{eq:pot}. Thus, by Lemma \ref{lem:Yedla1}, for any $\tilde{\gamma} >0$, for sufficiently large $\omega$ the fixed points $x_{\sfr}(\infty) := \lim_{t \to \infty} x_{\sfr}(t)$ satisfy:
\begin{equation} \label{eq:YedlaThm1_r}
    \max_{\sfr \in [\R]} \,  x_{\sfr}(\infty)   \leq \max \left( \argmin_{ x \in [0, \var(X)]} U(x; \delin)\right) +  \tilde{\gamma}.
\end{equation}

We claim that for $t \ge 0$:
\begin{align}
    \bar{x}_\sfr(t) \le x_\sfr(t), \quad \sfr \in [\R ].
    \label{eq:barx_xbound}
\end{align}
We prove the claim by induction. For $t=0$, we have $x_\sfr(0) = \var(X)$, and from \eqref{eq:GAMP-cr0} and \eqref{eq:BayesSEinit}: 
$$\bar{x}_\sfr(0) = \var(X) \sum_{\sfc=1}^{\Lc} W_{\sfr \sfc}
= \begin{cases}
    \var(X)  \frac{\sfr}{\omega} \, , &  1 \le \sfr \le (\omega-1), \\
    \var(X) \, , & \omega \le \sfr \le \Lambda,  \\
    \var(X)  \frac{\Lambda + \omega - \sfr}{\omega} \, , & \Lambda+1 \le \sfr \le \Lambda + \omega -1
\end{cases}
\le \,  x_\sfr(0).$$
Assume towards induction that 
\eqref{eq:barx_xbound} holds for some $t \ge 0$. Using this, we show that:
\begin{equation}
    y_{\sfr}(t+1) = g(x_\sfr(t)) \ge g(\bar{x}_\sfr(t)) \ge 
    \bar{y}_\sfr(t+1), \quad  r \in [\R].
    \label{eq:yrt_step}
\end{equation}
Then using \eqref{eq:yrt_step} in \eqref{eq:GAMP-cr2} and \eqref{eq:GAMP-cr2-mod}, and recalling that $f(y) = \mmse(-y)$ is non-increasing in $y$, we obtain that $ \bar{x}_\sfr(t+1) \le x_\sfr(t+1)$, for $ \sfr \in [\R ]$.

The first inequality in \eqref{eq:yrt_step} follows from the induction hypothesis  since $g(x)$ is increasing. The 
second inequality in \eqref{eq:yrt_step} uses \eqref{eq:GAMP-cr1}, \eqref{eq:gdef}, and the fact that for any $x >0$, we have
\begin{equation}
    \E_{P, Y}  \Big\{\var\left(Z \mid P, Y; \frac{x}{\delin}\right) \Big\} \ge 
\E_{P_{\sfr}, Y_{\sfr}}  \Big\{\var\left(Z_{\sfr} \mid P_{\sfr}, Y_\sfr; \frac{x}{\delin}\right) \Big\}, \quad \sfr \in [\R].
\label{eq:mmse_comp}
\end{equation}
To see \eqref{eq:mmse_comp}, we note that the left side is the MMSE of estimating $Z$ from the pair $(P,Y)$ where $ Z = P + G$ and $Y= \varphi(Z, \bar{\veps})$, with $P  \sim \normal(0, \E\{ Z^2\} - x/\delin)$, $G \sim \normal(0, x/\delin) \perp P$ and $\bar{\veps} \sim P_{\veps} \perp Z$. The right side is  the MMSE of estimating $Z_{\sfr}$ from the pair $(P_\sfr,Y_\sfr)$ where $Z_\sfr = P_\sfr  + G$ and $Y= \varphi(Z_\sfr, \bar{\veps})$ with $P_{\sfr
}  \sim \normal(0, \E\{ Z_{\sfr}^2\} - x/\delin)$, $G \sim \normal(0, x/\delin) \perp P_{\sfr}$ and  $\bar{\veps} \sim P_{\veps} \perp Z_{\sfr}$. We note that $Z$ and $Z_\sfr$ are both zero-mean Gaussian with variances $\E\{ Z^2\} \ge \E\{ Z_\sfr^2 \}$ for $\sfr \in [\Lr]$. Using arguments similar to those in the proof of \cite[Proposition 23]{barbier2019optimal}, it can be shown that $\E_{P, Y}  \Big\{\var\left(Z \mid P, Y; \frac{x}{\delin}\right) \Big\}$ is an increasing function of $\var(Z)$, and hence \eqref{eq:mmse_comp} holds. This completes the proof of the induction step, and the claim in \eqref{eq:barx_xbound} follows.

We now complete the proof by bounding the MSE of Bayes SC-GAMP in iteration $t$ in terms of $\max_{\sfr \in [\R]} \,  \bar{x}_{\sfr}(t)$, which can then be bounded via  \eqref{eq:barx_xbound} and \eqref{eq:YedlaThm1_r} in the limit as $t \to \infty$. Defining the shorthand $\psi_{\sfc}(t) := \mmse(\tq_{\sfc}(t)^{-1})$ and using \eqref{eq:GAMP-cr0} in \eqref{eq:mu_tau_p_Bayes}, for $\sfr \in [\R]$ we have:
\begin{align}
    &\bar{x}_{\sfr}(t)  = \sum_{\sfc=1}^{\C}
    W_{\sfr \sfc} \psi_{\sfc}(t) = \begin{cases}
   \frac{1}{\omega} \sum_{\sfc=1}^{\sfr} \psi_{\sfc}(t) \, , &  1 \le \sfr \le (\omega-1),  \\
  \frac{1}{\omega} \sum_{\sfc=\sfr-\omega+1}^{\sfr} \psi_{\sfc}(t) \, , & \omega \le \sfr \le \Lambda,  \\
\frac{1}{\omega} \sum_{\sfc=\sfr-\omega+1}^{\Lambda} \psi_{\sfc}(t) \, , & \Lambda+1 \le \sfr \le \Lambda + \omega -1.
\end{cases}
\label{eq:barxr_exp}
\end{align}
From \eqref{eq:MSE_exp}, the MSE of Bayes SC-GAMP converges almost surely as
 \begin{align}
   &  \lim_{n \to \infty} \,  \frac{\| \bx - \bargin^*(\bq(t))\|^2}{n}  = \frac{1}{\Lambda} \sum_{\sfc=1}^{\Lambda} \psi_\sfc(t) \nonumber \\
    & =\frac{1}{\Lambda} \left[ \sum_{\sfc=1}^{\omega} \psi_\sfc(t) + 
    \sum_{\sfc=\omega+1}^{2\omega} \psi_\sfc(t) + \ldots +
  \sum_{\sfc= \lceil \frac{ \Lambda}{\omega} -1 \rceil\omega+1}^{\Lambda} \psi_\sfc(t)
    \right],
    \label{eq:MSE_groups}
 \end{align}
where in the last line we have divided the sum
$\sum_{\sfc=1}^{\Lambda} \psi_\sfc(t)$
 into groups of $\omega$ non-intersecting consecutive terms. There are $\lceil\frac{\Lambda}{\omega}\rceil$ within the brackets, and using \eqref{eq:barxr_exp}, we have:
 \begin{align}
     \frac{1}{\Lambda} 
     \sum_{\sfc=1}^{\Lambda} \psi_\sfc(t)
     \leq  \frac{1}{\Lambda} \left\lceil{\frac{\Lambda}{\omega}} \right\rceil \omega \,  \max_{\sfr \in [\R]} \bar{x}_{\sfr}(t).
 \end{align}
Taking $t \to \infty$, and using \eqref{eq:barx_xbound} and \eqref{eq:YedlaThm1_r} we obtain that for sufficiently large $\omega$:
\begin{align}
     &  \limsup_{t \to \infty}   \frac{1}{\Lambda} 
     \sum_{\sfc=1}^{\Lambda} \psi_\sfc(t)
     \le \frac{\Lambda + \omega}{\Lambda} \max_{\sfr \in [\R]} \,  x_{\sfr}(t)  = \left( \max\left(\argmin_{x\in [0, \var(X)]} U(x; \delin)\right)+ \tilde{\gamma} \right) \frac{\Lambda+\omega}{\Lambda}.
     \label{eq:limsup_MSE}
\end{align}
The first line of \eqref{eq:MSE_groups} together with \eqref{eq:limsup_MSE} implies that for any $\gamma >0$, there exist finite $\omega_0, t_0$ such that for $\omega > \omega_0$ and $t >t_0$, we almost surely have:
    \begin{align}
    & \lim_{n\to\infty}\frac{1}{n}\| \bx \, - \, \bargin^*(\bq(t); \, t ) \|^2 
    \leq  \left( \max\left(\argmin_{x\in [0, \var(X)]} U(x; \delin)\right)+\gamma\right) \frac{\Lambda+\omega}{\Lambda}.
    \label{eq:omLam_MSE1}
\end{align}
This proves the first statement in Theorem \ref{thm:SC_FPs}.
\qed

\subsubsection{Proof of Part 2 of Theorem \ref{thm:SC_FPs}} For the i.i.d.~Gaussian sensing matrix, we have a $1 \times 1$ base matrix with $W_{11}=1$, and $\delin = \delta$.
From Corollary \ref{cor:BayesSE},  the state evolution equations for Bayes GAMP are, for $t \ge 1$:
\begin{align}
   &  \mu^q(t) = \mu^p(t) =1, \label{eq:muc_mur_Bayes_iid}\\
    &  \tp(t) = \frac{1}{\delta}  \mmse\big(\tq(t)^{-1} \big),  \label{eq:mu_tau_p_Bayes_2} \\
       & \frac{1}{ \tq_{\sfc}(t+1)} =  \frac{1}{\tp(t)}\Big(1  -     \frac{1}{\tp(t)}\E_{P, Y}  \left\{\var{(Z \mid P, Y ;  \, \tp(t))} \right\}\Big),
  \label{eq:mu_tau_q_Bayes_iid}
\end{align}
where 
$Y=\varphi(Z, \, \bar{\veps})$ with $\veps \sim P_{\bar{\veps}} \perp Z$, and   $(P, Z) \sim \normal(0, \Lambda(t))$ where
\begin{equation}
   \Lambda(t) = \begin{bmatrix}
   \frac{\E\{X^2 \}}{\delta}  -\tp(t)  &   \frac{\E\{X^2 \}}{\delta} - \tp(t) \\
    \frac{\E\{X^2 \}}{\delta} - \tp(t) &  \frac{\E\{X^2 \}}{\delta}
   \end{bmatrix}.
   \label{eq:ups_x_def_tau_p}
\end{equation}
Using this together with \eqref{eq:Bayes-SC-GAMP-MSE}, we have that the limiting MSE of Bayes GAMP for the i.i.d.~Gaussian case is 
\begin{align}
    \lim_{n \to \infty} \,  \frac{\| \bx - \bargin^*(\bq(t))\|^2}{n} 
    = \mmse(1/\tq(t)) =: x(t).
\end{align}
From \eqref{eq:mu_tau_p_Bayes_2}  we note that $x(t) = \delta \tp(t)$, and combining with \eqref{eq:mu_tau_q_Bayes_iid}, we can write the state evolution as a recursion in terms of the parameter $x(t)$:
\begin{align}
    x(t) = f(g( x(t-1)) ), \quad t \ge 1,
\end{align}
where the functions $f$ and $g$ are defined in \eqref{eq:fdef} and \eqref{eq:gdef}. From \eqref{eq:BayesSEinit}, the recursion is initialized with $x(0) = \var(X)$. Since $\mmse(s) < \var(X)$ for all $s >0$, we have $x(1) < x(0)$. This together with the fact that  $f$ and $g$ are strictly increasing implies that  $x(t) \in [0, \var(X)]$ is strictly decreasing and converges to to a limit $x(\infty) : = \lim_{t \to \infty} x(t)$. Since the recursion is initialized at $\var(X)$, the maximum value of $x(t)$, the limit is  given by the largest solution of the fixed point equation:
\begin{align}
    x = f(g(x)), \quad x \in [0, \var(X)].
    \label{eq:FP_fgx}
\end{align}

By construction (see \eqref{eq:PotFn_formula} and \eqref{eq:fdef}-\eqref{eq:gdef}), the potential function $U(x; \delta)$ is
\begin{align}
    U(x ; \delta) = xg(x) - \int_{0}^x g(z) dz - \int_{0}^{g(x)} f(z) dz,
\end{align}
Differentiating with respect to $x$, and setting $\frac{\partial U(x; \delta)}{\partial x}  =0$ yields the fixed point equation \eqref{eq:FP_fgx}. This completes the proof.  
\qed

\section{Discussion and future directions} \label{sec:conclusion}

In this work, we have shown that for a broad class of GLMs, a simple spatially coupled design combined with an efficient AMP algorithm  can achieve the asymptotic Bayes-optimal estimation error (of an i.i.d.~Gaussian design). Although we have assumed that the spatially coupled matrix has independent Gaussian entries, using the recent universality results in \cite{Wan22}, we expect that the SC-GAMP algorithm as well as Theorems \ref{thm:gen_sc_gamp_SE} and \ref{thm:SC_FPs} remain valid for spatial coupling with \emph{generalized white noise} matrices, a much larger class of sensing matrices with independent entries.  

We outline a few open questions and directions for future work. For noiseless compressed sensing, Donoho et al. \cite{donoho2013information} showed that with a spatially coupled design, perfect signal reconstruction is possible  whenever the sampling ratio is larger than $\bar{d}(P_X)$, the (upper) R{\'e}nyi information dimension. In particular, for discrete-priors, $\bar{d}(P_X)=0$, which implies that the signal of dimension $n$ can be efficiently reconstructed from $o(n)$ noiseless linear  measurements. A natural question is: what is the minimum sampling ratio required for  efficient  reconstruction in GLMs such as phase retrieval and rectified linear regression? The results in \cite[Fig. 1]{barbier2019optimal} for the Bernoulli-Gaussian prior suggest that the threshold is $\bar{d}(P_X)$ for phase retrieval, and $2\bar{d}(P_X)$ for rectified linear regression. Corollary \ref{corr:BayesOpt} provides a way to rigorously prove such results: one could analyze the corresponding potential function and show that the potential function $U(x; \delta)$ has a unique minimizer  at $x=0$ if and only if $\delta$ is above the postulated threshold.

Our model assumptions in Section \ref{subsec:model_assump} ensured that the signal could be reconstructed without sign ambiguity, and therefore SC-GAMP did not require any special initialization. For sign-symmetric output functions such as noiseless phase retrieval, this meant assuming a signal prior with nonzero mean. To extend SC-GAMP to zero-mean signal priors, one could initialize it with a spectral estimator \cite{mondelli2017fundamental,luo2019optimal}. Spectrally-initialized AMP for i.i.d.~Gaussian designs was analyzed in \cite{mondelli2021approximate}. However, our experiments  suggest that this spectral initialization performs poorly if directly applied to the spatially coupled setting.   This could be because the preprocessing function used for the spectral estimator is optimal for an i.i.d. design, but not  for the spatially coupled one. Designing effective spectral estimators for spatially coupled designs is an interesting direction for future work.

Corollary \ref{corr:BayesOpt} guarantees that the limiting MSE of SC-GAMP can be made arbitrarily close to the limiting MMSE, provided $\omega > \omega_0$ and $\Lambda$ is sufficiently large. The value of $\omega_0$ is unspecified, and as seen from the numerical results in Figures \ref{fig:sc-gamp-pr} and \ref{fig:sc-gamp-relu}, the gap from the Bayes-optimal curve can be significant if $\omega, \Lambda$ are not large enough.  A more refined analysis that clarifies how the MSE of SC-GAMP approaches the Bayes-optimal curve as $\omega, \Lambda$ increase would be valuable, and inform the  implementation of spatially coupled designs at finite dimensions. 

Another practically relevant research direction  is to investigate  spatially coupled designs composed of \emph{rotationally invariant} matrices,  a broad class which allows flexibility in choosing the spectrum of the underlying blocks.  AMP for linear and generalized linear models with  (uncoupled) rotationally invariant designs has been studied in a number of works recently, e.g., \cite{ma2017orthogonal, rangan2019vector, takeuchi2020rigorous, pandit2020inference, tian2022generalized, venkataramanan22estimation}.  In this setting, Ma et al. \cite{Ma_spectrum21} showed that the spectrum of the sensing matrix can significantly affect  estimation performance. For example, designs with `spikier' spectrum are better for models such as phase retrieval, whereas `flatter' spectra are better for others such as 1-bit compressed sensing. Moreover, rotationally invariant designs are  closely related to  DCT-based and Fourier-based designs which facilitate scalable and memory-efficient AMP implementations (see Appendix \ref{app:dct}).   Takeuchi \cite{takeuchi2023orthogonal} recently studied spatially coupled designs for linear models with rotationally  invariant row blocks.  Spatially coupled designs with the more general block-wise structure considered in this paper,  with each block being rotationally invariant,  could potentially yield   substantially smaller estimation error. Investigating such designs is an intriguing open question.

\newpage

\appendix

\section{Equivalence of potential functions and proof of Theorem \ref{thm:iidMMSE}} \label{app:pot_fun_equiv}

In this section, we describe the alternative potential function used in \cite{barbier2019optimal} and show how the limiting MMSE characterization  in Theorem \ref{thm:iidMMSE} can be obtained from the one in that paper.  As in  \cite[Section 4.1]{barbier2019optimal}, we will assume in this section that the GLM in  \eqref{eq:GLM_def} defined via $y = \varphi(z, \veps)$ induces a conditional density (or probability mass function) $\Pout(y  \mid z)$ so that 
\[ \Pout( \by  \mid \bz)  = \prod_{i=1}^m \, \Pout(y_i \mid z_i).   \]
  In addition, the MMSE characterization in \cite{barbier2019optimal} requires the following assumptions.

\textbf{(B1)} The components of $\bx$ and $\bveps$ are i.i.d.~according to $P_X$ and $P_{\bar{\veps}}$. The prior $P_X$ admits a finite third moment, and has at least two points in its support. 

\textbf{(B2)} The expectation $\E\left\{ \abs{ 
\varphi(G, \bar{\veps}) }^{2+\gamma}  \right\}$ is bounded for some $\gamma >0$, where $G \sim \normal (0,1) \, \perp \, \bar{\veps} \sim P_{\bar{\veps}}$.

\textbf{(B3)} For almost-all values of $\veps$ (w.r.t. the distribution $P_{\bar{\veps}}$, the function $z \mapsto \varphi(z, \veps)$ is continuous (Lebesgue) almost everywhere.

For brevity, in the following we let $\rho := \E\{ X^2\}$.

 \begin{defi}[Potential function of Barbier et al.  \cite{barbier2019optimal}]\label{def:pot_barb}
For $r>0$ and $q \in [ (\E\{ X\})^2, \, \rho]$, define the `replica-symmetric' potential function $\frs(q, r)$ as
\begin{align}
    \frs(q,r; \, \delta) &:= \psiin(r) \, + \, \delta \psiout(q; \, \delta) \, - \, \frac{rq}{2},  \quad \text{ where } \label{eq:frs_def} \\
    \psiin(r) & := \E_{X_0, G_0}\left\{ \ln \E_X\big\{ e^{rXX_0 + \sqrt{r}X G_0 - r X^2/2} \big\}  \right\}, \label{eq:psiin_def}\\
\psiout(q; \, \delta) & := \E_{V ,Y} \left\{
\ln \E_G\big\{ \Pout(Y \, | \,  \sqrt{q/\delta} \, V \, + \,  \sqrt{(\rho -q)/\delta} \, G ) \big\} 
 \right\}. \label{eq:psiout_def}   \end{align}
 In \eqref{eq:psiin_def}, $X \sim P_X, \, X_0 \sim P_X$, and $G_0 \sim \normal(0,1)$ are mutually independent.  In \eqref{eq:psiout_def}, $V, G \sim_{\normalfont \text{i.i.d.}} \normal(0,1)$  and $Y \sim \Pout(\cdot \mid \sqrt{q/\delta} \, V  +   \sqrt{(\rho -q)/\delta} \, G )$.
 \end{defi} 

Next, we define the set $\Gamma$ to be the critical points of $\frs$:
\begin{align}
\Gamma := \left\{ (q, r) \, \in \, [  (\E\{ X\})^2,  \, \rho] \times (\reals_+ \cup \{+\infty \}) 
\, \Big\vert \,   q= 2 \psiin^{\prime}(r), \ r = 2 \delta\psiout^{\prime}(q; \delta) \right\}.
\end{align}

We note that in \cite{barbier2019optimal}, the sensing matrix was scaled to have i.i.d.~zero-mean entries with variance $\frac{1}{n}$, whereas in our paper (including in Theorem \ref{thm:iidMMSE}) we assume that the entries have variance $\frac{1}{m}$. This is accounted for in the coefficients multiplying $V$ and $G$ in the definition of $\psiout$ in \eqref{eq:psiout_def}.

\begin{theorem}[MMSE for i.i.d.~Gaussian design \cite{barbier2019optimal}] Assume the setting of Theorem \ref{thm:iidMMSE}, along with the conditions \textbf{(B1)}, \textbf{(B2)}, \textbf{(B3)}, and that at least   one of \eqref{eq:sign_anti_sym} or \eqref{eq:exp_ne_0} holds. Further, assume that maximizer $(q^*, r^*) \in \Gamma $ of the potential $ \frs(q,r; \, \delta)$ is unique. Then,
 \begin{equation}
     \lim_{n \to \infty}  \, \frac{1}{n} \E\{ \| \bx - \E\{\bx \mid \Aiid, \, \by \} \|^2\} = \rho - q^*.
     \label{eq:MMMSE_q}
  \end{equation}
  \label{thm:iidMMSE_alt}
 \end{theorem}

We will show that the statement in Theorem \ref{thm:iidMMSE} is equivalent to that of Theorem \ref{thm:iidMMSE_alt}  by showing that the two potential functions $\pot(x; \delta)$  and $\frs(q,r; \, \delta)$ are equivalent in the following sense.

\begin{prop}
Suppose that $x^* \in [0, \var(X)]$ is the unique minimizer of $\pot(x; \delta)$ (in Definition \ref{def:pot}). Then,
\begin{align}
    q^*= \rho - x^*, \qquad r^* = \frac{\delta}{x^*} \ups(x^*;  \delta),
    \label{eq:qr_star}
\end{align}
lies in $\Gamma$ and is the unique maximizer of  $\frs(q,r; \, \delta)$. Here $\ups(x^*;  \delta)$ is defined in \eqref{eq:ups_x_def}. 
\label{prop:pot_fn_equiv}
\end{prop}
Theorem \ref{thm:iidMMSE} follows immediately from Theorem \ref{thm:iidMMSE_alt} together with Proposition \ref{prop:pot_fn_equiv}.

\paragraph{Proof of Proposition \ref{prop:pot_fn_equiv}.} 
For $X_0 \sim P_X \, \perp \, G_0 \sim \normal(0,1)$, it can be verified that  $\psiin(r)$ in \eqref{eq:psiin_def} satisfies
\[  H(\sqrt{r}X_0  +  G_0) = - \psiin(r) +  \frac{r \rho}{2}  + \frac{\ln (2 \pi e)}{2} \, ,  
\]
where $H(\cdot)$ is the Shannon entropy in nats. Using this in the definition of $\frs$ in \eqref{eq:frs_def}, we  have that
\begin{align}
    -2\frs(q,r; \, \delta)
    = 
    2I(X_0; \,\sqrt{r} X_0 + G_0) - (\rho-q)r - 2\delta \psiout(q; \delta), 
    \label{eq:equiv_pot_rep}
\end{align}
where we used $H(G_0) = \frac{1}{2} \ln(2\pi e)$. 
Using the mapping
\begin{align}
q \equiv q(x) := \rho -x, \quad r \equiv r(x) 
:= \frac{\delta}{x} \ups(x; \delta),
\end{align}
from the definition of $\pot(x; \delta)$ in \eqref{eq:pot}, we have that
\begin{align}
   \pot(x; \delta) = (q  - \rho)r \, + \, 
   \int_0^{\rho - q} \frac{\delta}{z} \ups(z;  \delta) \,  dz  \,  +  \, 2I\left(X_0 \, ;  \sqrt{r} \,  X + G_0 \right).
   \label{eq:Uxdel_alt0}
\end{align}
We claim that
\begin{align}
    2\delta \psiout(q; \delta)= 
    - \int_0^{\rho - q} \frac{\delta}{z} \ups(z;  \delta) \,  dz  \,  + \, C_\delta, \, 
    \label{eq:psiout_alt}
\end{align}
where $C_\delta$ is a constant that does not depend on $q$. Using \eqref{eq:psiout_alt} in \eqref{eq:Uxdel_alt0} gives
\begin{align}
       \pot(x; \delta) = (q  - \rho)r \, - \, 
 2\delta \psiout(q; \delta) + C_\delta   \,  +  \, 2I\left(X_0 \, ;  \sqrt{r} \,  X + G_0 \right).
   \label{eq:Uxdel_alt}
\end{align}
Comparing \eqref{eq:Uxdel_alt} with \eqref{eq:equiv_pot_rep}, we observe that $\pot(x; \delta)=-2\frs\left(q(x),r(x)\right)+C_\delta$. Therefore,  
if $x^* \in [0, \var(X)]$ is the unique minimizer of $U(x;  \delta)$, then 
\begin{align} (q^*, r^*) \equiv  (q(x^*), r(x^*)) = \argmax_{(q,r) \in \Gamma}  \, \frs(q,r; \, \delta).
\end{align}
Here we have used the fact that  $x^*$ being a minimizer  satisfies $\frac{\partial \pot(x; \delta)}{\partial x} =0$; therefore from the alternative representation of $\pot(x; \delta)$ in \eqref{eq:Uxdel_alt}, we have that $q(x^*)$ and $ r(x^*)$ satisfy $\frac{\partial \frs(q,r; \delta)}{\partial q} =0$ and $\frac{\partial \frs(q,r; \delta)}{\partial r} =0$, respectively. That is, $(q^*, r^*) \in \Gamma$.

We complete the proof by showing the claim in \eqref{eq:psiout_alt}. For $q \in [ (\E\{ X\})^2, \, \rho]$, we can write
\begin{align}
     2\delta \psiout(q; \delta) =  2\delta \int_{(\E\{ X\})^2}^{q} \psiout'(u; \delta)du \,  + \,  c_{\delta},
     \label{eq:psiout_integ}
\end{align}
for some constant $c_\delta$. We then use the following alternative representation for $\psiout^{\prime}(u; \delta)$ from \cite[Proposition 20]{barbier2019optimal}:
\begin{align}
    & \psiout^{\prime}(u; \delta) = \frac{1}{2(\rho - u)}\E\left\{ \big( \E\{ G \, |  \, Y, \, V \, ; (\rho-u)/\delta \} \big)^2 \right\},  \nonumber  \\
& \hspace{1in} \text{ where }  V, G \sim_{\normalfont\text{i.i.d.}} \normal(0,1) \  \text{ and  } Y \sim \Pout(\cdot \mid \sqrt{u/\delta} \, V  +   \sqrt{(\rho -u)/\delta} \, G ).
\label{eq:psiout_pr}
\end{align}
Letting $P= \sqrt{u/\delta} \, V$ and $Z= P + \sqrt{(\rho -u)/\delta} \, G $, we observe that $(P, Z)$ are jointly Gaussian with covariance matrix of the form in \eqref{eq:ups_x_def} with $x= (\rho-u)$. Moreover, we claim that
\begin{align}
 \E\left\{  \E\{ G \, |  \, Y, \, V \, ; (\rho-u)/\delta  \}^2 \right\}
    =  1 - \frac{\delta}{\rho-u} 
    \E_{P,Y}  \left\{\var\left(Z \, \middle| \, P, Y; \, (\rho-u)/\delta \right) \right\} = \ups(\rho-u ;  \delta),
    \label{eq:alt_gout_exp}
\end{align}
where $\var\left(Z \, \middle| \, P, Y; \, (\rho-u)/{\delta} \right)$  and $\ups(\rho-u ;  \delta)$ are as defined in  Definition \ref{def:pot}. To see \eqref{eq:alt_gout_exp},  we can start with the right side and write $Z= P + \sqrt{(\rho -u)/\delta} \, G$  to deduce that
\begin{align}
    & 1 - \frac{\delta}{\rho-u} 
    \E_{P,Y}  \left\{\var\left(Z \, \middle| \, P, Y; \, \frac{\rho-u}{\delta} \right) \right\}  = 1 - 
    \E_{P,Y}  \left\{\var\left(G \, \middle| \, P, Y; \, \frac{\rho-u}{\delta} \right) \right\} \nonumber \\
    & \quad  = 1 -  \left(\E\{ G^2\} \,  - \, 
    \E\left\{  \E\left\{ G \, |  \, Y, \, V \, ; \frac{\rho-u}{\delta}  \right\}^2  \right\}\right) = 
    \E\left\{  \E\left\{ G \, |  \, Y, \, V \, ; \frac{\rho-u}{\delta}  \right\}^2  \right\},
    \label{eq:alt_gout_exp1}
\end{align}
where we used $G\sim \normal(0,1)$. Using \eqref{eq:alt_gout_exp} in \eqref{eq:psiout_pr} and then in \eqref{eq:psiout_integ} we obtain
\begin{align}
     2\delta \psiout(q; \delta) & =  \int_{(\E\{ X\})^2}^{q} \frac{\delta}{\rho-u} \ups(\rho-u ;  \delta) du \,  + \,  c_{\delta} \nonumber \\
     & = \, 
     - \int_{\var(X)}^{\rho-q} \frac{\delta}{z} \ups(z ;  \delta) dz \, + \, c_\delta =
     - \int_{0}^{\rho-q} \frac{\delta}{z} \ups(z ;  \delta) dz \,
     + \, C_\delta,
\end{align}
which completes the proof of the claim in \eqref{eq:psiout_alt}. The proposition follows. \qed

\section{Proof of Corollary \ref{cor:BayesSE}} \label{app:corr_proof}

The proof uses the definitions of $\ginbayes$ and $\goutbayes$ in the state evolution equations \eqref{eq:tq_t1_def}-\eqref{eq:Lambda22}, along with Stein's lemma and elementary properties of conditional expectation. We first  prove that $\mu^{p}_{\sfr}(t)=1$ for all $\sfr$.  Using the definition of $\ginbayes$ from \eqref{eq:ginbayes} and the tower property of conditional expectation we have that:
\begin{align}
    \E\{ X \, \ginbayes(Q_{\sfc}(t), \sfc \, ;  \, t)\} ) = \E\{ \E\{ X \ginbayes(Q_{\sfc}(t), \sfc \, ;  \, t)  \mid Q_{\sfc}(t) \} \} = \E\big\{  \ginbayes(Q_{\sfc}(t), \sfc \, ;  \, t) ^2 \big\}, \quad \sfc \in [\C].
    \label{eq:EXginbayes}
\end{align}
Using this in \eqref{eq:Lambda11}, \eqref{eq:Lambda12},  and  \eqref{eq:mup_t1_def}, we obtain that $\mu^{p}_{\sfr}(t)=1$ for $r \in [\R]$. 

To prove the other results, we show that for $\sfr \in [\R]$ and $t \ge 0$:
\begin{align}
\E\{ \goutbayes( P_{\sfr}(t), \, \varphi(Z_{\sfr}, \bar{\veps}), \sfr \,  ; \, t )^2 \} 
& = - \E\{ \goutbayes'( P_{\sfr}(t), \, \varphi(Z_{\sfr}, \bar{\veps}), \sfr \,  ; \, t ) \} \nonumber \\
& = \E\{ \partial_z \goutbayes( P_{\sfr}(t), \, \varphi(Z_{\sfr}, \bar{\veps}), \sfr \,  ; \, t ) \},
 \label{eq:Epartialgout}
\end{align}
where we recall that 
$\goutbayes'(p, \varphi(z, \veps), \sfr \, ; \, t)$ denotes the derivative with respect to $p$, and  $\partial_z \, \goutbayes(p, \varphi(z, \veps), \sfr \, ; \, t)$ the derivative with respect to $z$. Using the second equality of \eqref{eq:Epartialgout} in \eqref{eq:muq_t1_def} yields $\mu_{\sfc}^q(t)=1$ for $t \ge 1$. And using the first equality of \eqref{eq:Epartialgout} in \eqref{eq:tq_t1_def}, and recalling that $Y_{\sfr} = \varphi(Z_{\sfr}, \bar{\veps})$, we obtain:
\begin{equation}
    \frac{1}{\tq_{\sfc}(t+1)} = -  \sum_{\sfr =1}^{\R} W_{\sfr \sfc} \, \E\{ \goutbayes'( P_{\sfr}(t), \, Y_{\sfr}, \sfr \,  ; \, t )\}, \ \text{ for } t \ge 0.
    \label{eq:tq_t1_Bayes}
\end{equation}
The derivative $\goutbayes'$ is given by \cite[Eq. (37)]{rangan2010generalized}:
\begin{align}
    -\goutbayes'(P_{\sfr}(t),  Y_{\sfr}, \sfr \, ;  \, t)
= \frac{1}{\tp_{\sfr}(t)}\left( 1 - \frac{\var\{ Z_{\sfr} \mid P_{\sfr}(t), \, Y_{\sfr} \}}{\tp_{\sfr}(t)} \right), \ \text{ for } \sfr \in [\R].
\label{eq:goutB_der}
\end{align}
Using \eqref{eq:goutB_der} in \eqref{eq:tq_t1_Bayes} gives the result in \eqref{eq:mu_tau_q_Bayes} for $t \ge 1$.

For $\tp_{\sfr}(t)$, we use the formula in \eqref{eq:mup_t1_def} and note that $[\Lambda_{\sfr}(t)]_{11} = [\Lambda_{\sfr}(t)]_{12}$, by \eqref{eq:EXginbayes}. This gives:
\begin{align}
    \tp_{\sfr}(t) & =  \frac{\E\{ X^2 \}}{\delin} \sum_{\sfc \in [\C]} W_{\sfr, \sfc}   \,  - \, [\Lambda_{\sfr}(t)]_{11}, \ \text{ for } \sfr \in [\Lr].
       \label{eq:tp_r_bayes0}
\end{align}
Using the formula for $[\Lambda_{\sfr}(t)]_{11}$ in  \eqref{eq:Lambda11} and the definition of $\ginbayes$ from \eqref{eq:ginbayes}, \eqref{eq:tp_r_bayes0} gives:
\begin{align}
    \tp_{\sfr}(t) =  \frac{1}{\delin} \sum_{\sfc \in [\C]} W_{\sfr, \sfc} \left[ \E\{ X^2 \} \,  - \, \E\big\{ (\E\{X \mid Q_{\sfc}(t)\})^2 \big\} \right] = \frac{1}{\delin} \sum_{\sfc \in [\C]} W_{\sfr, \sfc} \mmse(\tq_{\sfc}(t)^{-1}).
    \label{eq:tp_r_bayes}
\end{align}
The last equality in \eqref{eq:tp_r_bayes} follows from the representation of $Q_{\sfc}(t)$ in \eqref{eq:PtQt_def} by noting that $\mu^{q}_{\sfc}(t)=1$. 

The identity $\alpha_\sfc^q(t+1) = \tq_\sfc(t+1)$, for $t \ge 0$, can be obtained by using \eqref{eq:goutB_der} in the definition \eqref{eq:alpha_qtq_def}. Substituting $\alpha_\sfc^q(t+1)=\tq_\sfc(t+1)$ in \eqref{eq:alpha_pt_def} gives 
\begin{align}
    \alpha_\sfr^p(t+1) =\frac{1}{\delin}\sum_{\sfc=1}^\C W_{\sfr,\sfc}\, \tq_\sfc(t+1)\,\E\left\{\ginbayes^\prime\left(Q_{\sfc}(t+1), \sfc\,;\,t+1\right)\right\},\quad \text{for }\sfr\in[\R],\label{eq:alphap-exp}
\end{align}
Recalling the definition of $\ginbayes$ in \eqref{eq:ginbayes}, its derivative  can be written as \cite[Eq. (32)]{rangan2010generalized}:
\begin{equation}
    \ginbayes^\prime\left(Q_{\sfc}(t+1), \sfc\,;\,t+1\right) = \frac{\var\left(X \mid Q_{\sfc}(t+1), \sfc\,;\, t+1\right)}{\tq_c(t+1)}. \label{eq:ginderiv-var}
\end{equation}
Moreover, as in \eqref{eq:tp_r_bayes} we have that
\begin{equation}
    \mmse\left(\tq_\sfc(t+1)^{-1}\right)= \E\big\{ (X - \E\{X \mid Q_{\sfc}(t+1)\})^2 \big\} = \E\left\{\var\left(X \mid Q_{\sfc}(t+1), \sfc\,;\,t\right)\right\}. \label{eq:mmse-var}
\end{equation}
Using \eqref{eq:ginderiv-var} and \eqref{eq:mmse-var} in \eqref{eq:alphap-exp}, and comparing with the expression in \eqref{eq:tp_r_bayes} gives $\alpha^p_\sfr(t+1) = \tp_\sfr(t+1)$.

It remains to prove the two claims in \eqref{eq:Epartialgout}. For  the first claim, from the definition of $\goutbayes$ in \eqref{eq:goutbayes}, we have that
\begin{align}
    \E\{ \goutbayes( P_{\sfr}(t), \, Y_{\sfr}, \sfr \,  ; \, t )^2 \} &  = \frac{1}{\left( \tp_{\sfr}(t) \right)^2}\left(
    \E\{ P_{\sfr}(t) ^2\} + 
    \E\big\{ \E\{ Z_{\sfr} \mid P_{\sfr}(t), \, Y_{\sfr} \}^2\big\} - 2 \E\{ P_{\sfr}(t) \E\{ Z_{\sfr} \mid P_{\sfr}(t), Y_{\sfr} \}  \right) \nonumber \\
    & \stackrel{(\rm{i})}{=}\frac{1}{\left( \tp_{\sfr}(t) \right)^2}\left(  \E\{ P_{\sfr}(t) ^2\} + \E\big\{ \E\{ Z_{\sfr} \mid P_{\sfr}(t), \, Y_{\sfr} \}^2\big\}  -
    2 \E\{  Z_{\sfr} P_{\sfr}(t) \} \right) \nonumber \\
    & \stackrel{(\rm{ii})}{=} \frac{1}{\tp_{\sfr}(t)} \left(  1 - \frac{\E\left\{\var(Z_{\sfr} \mid P_{\sfr}(t), Y_{\sfr})\right\}}{\tp_{\sfr}(t)} \right). \label{eq:goutBvar}
\end{align}
Here, the equality $(\rm{i})$ is obtained using the tower property of expectation, and $(\rm{ii})$ using the covariance of $(P_{\sfr}(t), Z_{\sfr})$ in \eqref{eq:PZ_dist_alt}, which along with  $\mu^p_{\sfr}(t)=1$, yields $ Z_{\sfr} = P_{\sfr}(t)  +  G^p_\sfr(t)$, implying that:
\[ 
\E\{ P_{\sfr}(t)^2 \} = \E\{ P_{\sfr}(t)  Z_{\sfr} \} = 
\E\{ Z_{\sfr}^2 \} - \tp_{\sfr}(t).
\]
Comparing \eqref{eq:goutB_der} and \eqref{eq:goutBvar}, we obtain the first equality in \eqref{eq:Epartialgout}. Finally, we show that the first and third terms in \eqref{eq:Epartialgout} are equal. Using
the formula for $\goutbayes$ in \eqref{eq:goutbayes}, for $\sfr \in [\Lr]$ we write:
\begin{align}
    \E\{ \goutbayes( P_{\sfr}(t), \, \varphi(Z_{\sfr}, \bar{\veps}), \sfr \,  ; \, t )^2 \} & = 
    \frac{1}{\tp_{\sfr}(t)} \E\left\{ (\E\{ Z_{\sfr} \mid P_{\sfr}(t), \, Y_{\sfr} \} \, - \, P_{\sfr}(t)) \cdot  \goutbayes( P_{\sfr}(t), \, Y_{\sfr}, \sfr \,  ; \, t ) \right\} \nonumber  \\
    & =  \frac{1}{\tp_{\sfr}(t)} \E\left\{ (Z_\sfr - P_{\sfr}(t)) \cdot  \goutbayes( P_{\sfr}(t), \, Y_{\sfr}, \sfr \,  ; \, t ) \right\}.
    \label{eq:Egout2_alt}
\end{align}
where the second equality holds due to the tower property of conditional expectation. Recalling that 
$Z_\sfr =  P_{\sfr}(t) + G^p_\sfr(t)$,
with $G^p_\sfr(t) \sim \normal(0, \tp_{\sfr}(t))$ independent of $P_{\sfr}(t) \sim \normal(0, \E\{ Z_{\sfr}^2\}-\tp_{\sfr}(t))$, the expectation on right side of \eqref{eq:Egout2_alt} can be written as
\begin{align}
 &  \E\left\{ (Z_\sfr - P_{\sfr}(t)) \cdot  \goutbayes( P_{\sfr}(t), \, Y_{\sfr}, \sfr \,  ; \, t )
    \right\}  = 
\E\left\{ G^p_\sfr(t) \,  \goutbayes( P_{\sfr}(t), \, \varphi(P_{\sfr}(t) + G^p_\sfr(t), \,  \bar{\veps}), \sfr \,  ; \, t ) 
    \right\} \nonumber \\ 
    &  \hspace{1in} = \E\left\{  \, \E\left\{ G^p_\sfr(t) \,  \goutbayes( P_{\sfr}(t), \, \varphi(P_{\sfr}(t) + G^p_\sfr(t), \,  \bar{\veps}), \sfr \,  ; \, t )  \mid P_{\sfr}(t)
    \right\} \right\}.
    \label{eq:inner_exp0}
\end{align}
We now apply Stein's lemma (stated as Lemma \ref{lem:stein} below) to the inner expectation, noting that $P_{\sfr}(t) \perp G^p_\sfr(t)$, to obtain
\begin{align}
   &  \E\left\{ G^p_\sfr(t) \,  \goutbayes( P_{\sfr}(t), \, \varphi(P_{\sfr}(t) + G^p_\sfr(t), \,  \bar{\veps}), \sfr \,  ; \, t )  \mid P_{\sfr}(t)
    \right\} \nonumber \\
    & = \tp_{\sfr}(t) \, \E\{ \partial_z \goutbayes( P_{\sfr}(t), \, \varphi(P_{\sfr}(t) + G^p_\sfr(t), \,  \bar{\veps}), \sfr \,  ; \, t )  \mid P_{\sfr}(t)
    \}.
    \label{eq:stein_app}
\end{align}
Substituting \eqref{eq:stein_app} in \eqref{eq:inner_exp0} and using the result in  \eqref{eq:Egout2_alt}, we obtain
\begin{equation}
  \E\{ \goutbayes( P_{\sfr}(t), \, \varphi(Z_{\sfr}, \bar{\veps}), \sfr \,  ; \, t )^2 \}  
  = \E\{ \partial_z \goutbayes( P_{\sfr}(t), \, \varphi(P_{\sfr}(t) + G^p_\sfr(t), \,  \bar{\veps}), \sfr \,  ; \, t ) \}.
\end{equation}
where we used $Y_\sfr = \varphi(Z_\sfr, \,  \bar{\veps} ) = \varphi(P_{\sfr}(t) + G^p_\sfr(t), \,  \bar{\veps})$. This completes the proof of \eqref{eq:Epartialgout}, and  the corollary follows.  \qed

\begin{lemma}[Stein's lemma \cite{stein1972bound}]
Let $U$ be a Gaussian random variable  with zero mean,
and let $f: \reals \to \reals$ be any function  such that $\E\{ U f(U)\}$ and $\E\{f'(U)\}$ exist. Then,
\[
\E\{ U f(U)\} = \var(U) 
\, \E\{ f'(U)  \}.
\]
\label{lem:stein}
\end{lemma}

\section{Computational details} \label{app:comp_details}
\subsection{Computing the state evolution parameters}\label{app:comp_SE}
For iteration $t$, the functions $\bargin$ and $\bargout$ in the SC-GAMP decoder in \eqref{eq:SC-GAMP1} and \eqref{eq:SC-GAMP2}, as well as  vectors $\bttq(t+1)$ and $\bttp(t+1)$, will  in general depend on the state evolution parameters $\tp_\sfr(t), \mu^p_\sfr(t), \tq_\sfc(t), \tq_\sfc(t+1), \mu^q_\sfc(t)$ and $\mu^q_\sfc(t+1)$, for $\sfr \in [\R], \sfc \in [\C]$.  

For Bayes SC-GAMP, from Corollary \ref{cor:BayesSE} we have  $\mu^p_\sfr(t)=\mu^q_\sfc(t)=1$, for $t\geq 1$ and $\sfr \in [\R], \sfc \in [\C]$. Instead of computing the parameters $\tq_\sfc(t)$ and $\tp_\sfr(t)$  offline  through the deterministic recursion in \eqref{eq:mu_tau_q_Bayes}-\eqref{eq:mu_tau_p_Bayes}, in our implementation they are estimated at runtime using the outputs $\bp(t)$ and $\bq(t)$ of the SC-GAMP decoder at each iteration $t$. These estimates, denoted by $\htq_\sfc(t)$ and $\htp_\sfr(t)$, are computed as follows, starting from the initialization $\tp_\sfr(0)$ in \eqref{eq:BayesSEinit}. For $t \ge 0$:
\begin{align}
    \htq_\sfc(t+1) & = \left( \sum_{\sfr=1}^\R W_{\sfr,\sfc}\frac{1}{m/\R}\sum_{i\in \mc{I}_\sfr}\left(-\goutbayes^\prime\left({{p}}_i({t}), y_i, \sfr\,;\,t\right)\right) \right)^{-1},\quad \text{for }\sfc\in[\C], \label{eq:tauq-data}\\
    \htp_\sfr(t+1) & =\frac{1}{\delin}\sum_{\sfc=1}^\C \left(W_{\sfr,\sfc}\, \tq_\sfc(t+1)\frac{1}{n/\C}\sum_{j\in \mc{J}_\sfc}\left(\ginbayes^\prime\left({{q}}_j({t+1}), \sfc\,;\,t+1\right)\right)\right),\quad \text{for }\sfr\in[\R].\label{eq:taup-data}
\end{align}
 The estimates above are derived by replacing expectations by empirical averages. Indeed, using \eqref{eq:goutB_der}, \eqref{eq:mu_tau_q_Bayes} can be rewritten as
\begin{equation}
  \tq_\sfc(t+1)  = \left(\sum_{\sfr=1}^\R W_{\sfr,\sfc}\E\left\{-\goutbayes^\prime(P_\sfr(t), Y_r, \sfr\,;\,t)\right\} \right)^{-1}. \label{eq:tauq-exp}
\end{equation}
Estimating $\E\{-\goutbayes^\prime(P_\sfr(t), Y_r, \sfr\,;\,t)$ via the empirical average over the $\sfr$th row block, for each $r \in [\Lr]$, gives  \eqref{eq:tauq-data}. To derive \eqref{eq:taup-data}, we recall from Corollary \ref{cor:BayesSE} that $\alpha^p_\sfr(t+1) = \tp_\sfr(t+1)$, and use the expression for $\alpha^p_\sfr(t+1)$ in
\eqref{eq:alphap-exp},  replacing $\E\left\{\ginbayes^\prime\left(Q_{\sfc}(t+1), \sfc\,;\,t+1\right)\right\}$ by the empirical average over the $\sfc$th column block, for $\sfc \in [\Lc]$.

\subsection{Explicit expressions for $\ginbayes$, $\goutbayes$ and potential function} \label{app:comp_GLM}
\subsubsection{Phase Retrieval}
In the phase retrieval model, we have
\begin{equation}
    y_i=\left(\bA\bx\right)_i^2+w_i=z_i^2+w_i,\quad \text{ for }i\in[m],
    \label{eq:noisypr}
\end{equation}
where $w_i\sim \normal(0,\sigma^2)$. 
Under this model,  the conditional  density of $Y$ given $Z$ is
\begin{equation}
 \Pout\left(Y=y|Z=z\right)=\frac{1}{\sqrt{2\pi \sigma^2}}\exp\left(\frac{-\left(y-z^2\right)^2}{2\sigma^2}\right). \label{eq:pyz-pr}
\end{equation}

\paragraph{Output function $\goutbayes$:} 
The expression for $\goutbayes$ is obtained by computing the conditional expectation $\E\left\{ Z_{\sfr}|P_\sfr(t),Y_\sfr \right\}$ in \eqref{eq:goutbayes},  using the joint distribution of $(Z_\sfr, P_\sfr(t))$ given by \eqref{eq:PZ_dist_alt}, and conditional distribution of  $Y_\sfr$ and $Z_\sfr$ given by \eqref{eq:pyz-pr}.
For noiseless phase retrieval, we first obtain the expression for $\E\left\{ Z_{\sfr}|P_\sfr(t),Y_\sfr \right\}$  for $\sigma >0$, and then take the limit  $\sigma \to 0$. This gives the following expression for $\goutbayes$:
\begin{equation}
    \goutbayes\left({p}, {y}; \tp\right)
    =\frac{1}{\tp}\left(\sqrt{y}\tanh{\left(\tfrac{p\sqrt{y}}{\tp}\right)}H(y)-p\right), \label{eq:gout-pr}
\end{equation}
where $H(y)$ is the Heaviside step function defined as
\begin{equation} \label{eq:heaviside}
    H(y)=\begin{cases}
    1\quad\text{ for }y>0,\\
    0\quad\text{ for }y\leq0.
    \end{cases}
\end{equation}
The  derivative of $\goutbayes$ (with respect to its first argument) can be obtained directly by differentiating \eqref{eq:gout-pr} and is given by
\begin{equation}
    \goutbayes^{\prime}\left({p}, {y}; \tp\right)=\frac{\partial}{\partial p}\goutbayes\left({p}, {y}; \tp\right)=\frac{-1}{\tp}\left(1 - \frac{y}{\tp}\left(1-\tanh^2{\left(\tfrac{p\sqrt{y}}{\tp}\right)}\right)H(y)\right).
    \label{eq:gout_deriv_pr}
\end{equation}
\paragraph{Input function $\ginbayes$:} The Bayes-optimal denoising function $\ginbayes$, defined in \eqref{eq:ginbayes} depends only on the signal distribution $p_X$.

For the numerical simulations presented in Sec. \ref{sec:numerical_results}, a non-zero mean signal distribution  was chosen for phase retrieval in order to avoid meeting both sign-symmetry conditions in \eqref{eq:sign_anti_sym}, \eqref{eq:exp_ne_0}. The signal distribution was of the form
\begin{equation*}
    p_X = \alpha \delta_a + (1-\alpha) \delta_{-a},
\end{equation*}
where $a=\sqrt{1/\left(1-\left(2\alpha-1\right)^2\right)}$ is chosen such that $\var(X)=1$. For our experiments, we used $\alpha =0.6$.
Under this prior,  we have
\begin{equation}
\begin{split}
    \ginbayes(q, t) 
   &  = \E\{ X \mid X+ \sqrt{\tq(t)}G_0 =q \}   \\
    %
    &=\dfrac{\alpha a\ \phi\left(\tfrac{q-a}{\sqrt{\tq(t)}}\right) - (1-\alpha) a \ \phi\left(\tfrac{q+a}{\sqrt{\tq(t)}}\right)}{ \alpha \ \phi\left(\tfrac{q-a}{\sqrt{\tq(t)}}\right) + (1- \alpha) \ \phi\left(\tfrac{q+a}{\sqrt{\tq(t)}}\right)},
\end{split}\label{eq:ginbayes-pr}
\end{equation}
where $G_0\sim\normal(0,1)$ and $\phi(x)=\left({1}/{\sqrt{2\pi}}\right)e^{-\tfrac{x^2}{2}}$ is the standard normal density. 
\paragraph{Potential function.}
From \eqref{eq:pot}, \eqref{eq:ups_x} we observe that the potential function is expressed in terms of $\ups(x;  \delta) =1 - \frac{\delta}{x} \, \E_{P,Y}  \left\{\var\left(Z\middle|P, Y; \, \frac{x}{\delta} \right)\right\}$. From \eqref{eq:goutB_der},  $\var\left(Z\middle|P, Y; \, \frac{x}{\delta} \right)$ can be expressed in terms of $\goutbayes^\prime$ as derived in \eqref{eq:gout_deriv_pr}, with $\tp=x/\delta$. For phase retrieval, this results in
\begin{equation}
    \ups(x;  \delta)= 1 -  \dfrac{\delta}{x}\mathbb{E}_{P, Y} \left\{ {y}\left(1-\tanh^2{\left(\tfrac{p\sqrt{y}}{x/\delta}\right)}\right)H(y)\right\}. \label{eq:ups_pr}
\end{equation}
The potential function then takes the form in \eqref{eq:pot}, using the expression in \eqref{eq:ups_pr} for $\ups(x;  \delta)$.

\subsubsection{Rectified Linear Regression}
In the noiseless rectified linear regression model,  we have 
\begin{equation}
   y_i =\text{max}\left((\bA\bx)_i, \, 0 \right) = \text{max}\left(z_i, \, 0 \right) , \quad \text{ for } i \in [m]. 
\end{equation} 
Therefore, the output function, often referred to as ReLU (rectified linear unit), acts as a thresholding function  that maps negative $z_i$ values to 0.
This implies that the conditional density of $Y$ given $Z$ can be written as
\begin{equation}
    {\Pout}\left(Y=y|Z=z\right)=\mathbbm{1}\left\{ y=z\right\}\cdot \mathbbm{1}\left\{ z>0\right\}\  + \ \mathbbm{1}\left\{ y=0\right\}\cdot \mathbbm{1}\left\{ z\leq0\right\}.
\end{equation}

\paragraph{Output function $\goutbayes$:} To compute  $\goutbayes$ using \eqref{eq:goutbayes},  we consider two cases for $y$ separately. For $y>0$, since $y=z$, we have $\E\{Z_\sfr | P_\sfr(t)=p,Y_\sfr=y\}=y$.
For $y=0$, which corresponds to all  negative values of $z$, the expectation $\E\{Z_\sfr | P_\sfr(t)=p,Y_\sfr=y\}$ can be computed explicitly in terms of the standard normal density and distribution functions.
Overall, we can express $\goutbayes$ as
\begin{equation} \label{eq:gout-relu}
    \goutbayes\left({p}, y; \tp\right)=\begin{cases}\quad\frac{-1}{\sqrt{\tp}}\frac{\phi\left(\frac{p}{\sqrt{\tp}}\right)}{\Phi\left(\frac{-p}{\sqrt{\tp}}\right)}\quad&\text{if }y=0,\\ \\\quad\frac{1}{\tp}\left(y-p\right)\quad&\text{ if }y>0,\end{cases}
\end{equation}
where $\phi$ and $\Phi$ are the standard normal density and distribution functions.
The derivative $\goutbayes^{\prime}$, obtained by differentiating Equation \eqref{eq:gout-relu}, is given by
\begin{equation}
    \frac{\partial}{\partial p} \goutbayes\left({p}, y; \tp\right) =\begin{cases}\quad\frac{-1}{\tp}\left(\frac{-\frac{p}{\sqrt{\tp}}\cdot\phi\left(\frac{p}{\sqrt{\tp}}\right)\Phi\left(\frac{-p}{\sqrt{\tp}}\right) + \phi^2\left(\frac{p}{\sqrt{\tp}}\right)}{\Phi^2\left(\frac{-p}{\sqrt{\tp}}\right)}\right)&\quad\text{if }y=0,\\ \\\quad\frac{-1}{\tp}&\quad\text{ if }y>0.\end{cases}
    \label{eq:dgout-ReLU}
\end{equation}

\paragraph{Input function $\ginbayes$:} 
For the numerical simulations presented in Figure  \ref{fig:sc-gamp-relu}, we use a  sparse prior, where each signal entry takes one of three values $\{-b, 0, b\}$. The signal distribution was of the form
\begin{equation*}
    p_X =  (1-\alpha) \delta_0 + \frac{\alpha}{2} \delta_{-b} +  \frac{\alpha}{2} \delta_{b},
\end{equation*}
where $b=\sqrt{1/\alpha}$ is chosen such that $\var(X)=1$. For our experiments, we used $\alpha =0.5$. Under this prior, we have
\begin{equation}
\begin{split}
    \ginbayes(q, t) &= \E\{ X \mid X+ \sqrt{\tq(t)}G_0 =q \}  \\
    %
    &=\dfrac{\frac{\alpha}{2}b\ \phi\left(\tfrac{q-b}{\sqrt{\tq(t)}}\right) - \frac{\alpha}{2}b\ \phi\left(\tfrac{q+b}{\sqrt{\tq(t)}}\right)}{(1-\alpha)\ \phi\left(\tfrac{q}{\sqrt{\tq(t)}}\right) + \frac{\alpha}{2}\ \phi\left(\tfrac{q-b}{\sqrt{\tq(t)}}\right) + \frac{\alpha}{2}\ \phi\left(\tfrac{q+b}{\sqrt{\tq(t)}}\right)}.
\end{split}\label{eq:ginbayes-relu}
\end{equation}

\paragraph{Potential function.}
From \eqref{eq:pot}, \eqref{eq:ups_x} we observe that the potential function is expressed in terms of $\ups(x;  \delta) =1 - \frac{\delta}{x} \, \E_{P,Y}  \left\{\var\left(Z\middle|P, Y; \, \frac{x}{\delta} \right)\right\}$. From \eqref{eq:goutB_der},  $\var\left(Z\middle|P, Y; \, \frac{x}{\delta} \right)$ can be expressed in terms of $\goutbayes^\prime$ as defined in \eqref{eq:dgout-ReLU} with $\tp=x/\delta$. For rectified linear regression, this results in
\begin{equation}
    \ups(x;  \delta)= \mathbb{E}_{P, Y} \left\{\mathbbm{1}\{y=0\}\left(\frac{-\frac{p}{\sqrt{x/\delta}}\cdot\phi\left(\frac{p}{\sqrt{x/\delta}}\right)\Phi\left(\frac{-p}{\sqrt{x/\delta}}\right) + \phi^2\left(\frac{p}{\sqrt{x/\delta}}\right)}{\Phi^2\left(\frac{-p}{\sqrt{x/\delta}}\right)}\right) + \mathbbm{1}\{y>0\}\right\}. \label{eq:ups_relu}
\end{equation}
The potential function then takes the form in \eqref{eq:pot}, using the expression in \eqref{eq:ups_relu} for $\ups(x;  \delta)$.

\subsection{Potential function} \label{app:pot}
The global minimizer (orange) and largest stationary point (green) curves in Figures \ref{fig:sc-gamp-pr} and \ref{fig:sc-gamp-relu} are generated by finding the stationary points of the potential function $\pot(x; \delta)$ defined in \eqref{eq:pot} as
 \begin{equation} 
 \begin{split}
    &  U(x;  \delta) :=  - {\delta} \ups(x;  \delta)   \, +  \, \int_0^x \frac{\delta}{z} \ups(z;  \delta) \,  dz  \,  +  \, 2I\left(X \, ;  \sqrt{( \delta/x) \ups(x; \delta)} \,  X + G_0 \right).
\end{split} \label{eq:pot-comp-det}
 \end{equation} 
For each sampling ratio $\delta$, the potential function is evaluated at 200 data points between 0.005 and 0.995. Then, the stationary points and global minimizer are found by comparing these 200 evaluations. The full range $x\in[0, 1]$ is truncated to avoid computation errors. Furthermore, the lower limit of the integral in \eqref{eq:pot-comp-det} is set to 0.001 (instead of 0), again to avoid computational instabilities. In addition, the expression for $\ups(x;  \delta)$ includes the expectation $\E_{P,Y}  \left\{\var\left(Z\middle|P, Y; \, \frac{x}{\delta} \right)\right\}$, which corresponds to a double integral over the support of $P$ and $Y$. All integrals are computed using numerical integration rather than Monte Carlo methods, since the potential function must be smooth in order to reliably find its stationary points. 

\subsection{DCT-based Implementation} \label{app:dct}
When computing SC-GAMP for larger base matrix parameters such as $(\omega=20,\Lambda=200)$, the signal dimension $n$ in the numerical simulations needs to be very large (greater than 20,000)  to ensure that the sizes of the column blocks and the row blocks   ($n/\C$ and $m/\R$, respectively) are sufficiently large. However, this is not computationally feasible while explicitly constructing the sensing matrix $\bA$ due to the high memory requirements of the matrix-vector operations. This problem can be averted by using a discrete cosine transform (DCT) based construction of $\bA$. 

The DCT is used to generate the sensing matrix $\bA$ by randomly picking $m/\R$ rows and $n/\C$ columns from a larger DCT coefficient matrix for each block in $\bA$. For $\sfr\in [\R]$ and $\sfc\in [\C]$, the entries in the ($\sfr, \sfc$)th block are then normalized by $\sqrt{W_{\sfr, \sfc}/\left(m/\Lr\right)}$ such that each column of $\bA$ has norm 1. Matrix-vector multiplications with this DCT-based sensing matrix can be computed using the Fast Fourier Transform (FFT), reducing the memory requirements of the algorithm and the decoding complexity of the matrix-vector multiplications from $O(n^2)$ to $O(n \log n)$. This approach has been used before in the context of sparse regression codes \cite{rush2017capacity,barbier2017approximate} and further details about the implementation can be found in \cite[Sec. 2.5]{hsieh2021thesis}.

It is important to clarify that the entries in the DCT-based sensing matrix are not i.i.d.~Gaussian. Therefore, the potential function and state evolution results presented in this paper will not apply to the DCT-based design. Nevertheless, as shown by the numerical results in Figure \ref{fig:sc-gamp-relu}, the DCT-based design provides a practical and effective SC-GAMP implementation for large spatial coupling parameters. Furthermore, for large dimensions, the DCT-based design was found to have a similar error performance to the i.i.d.~Gaussian design.
Obtaining a rigorous state evolution characterization for DCT-based spatially coupled designs is an open question.

\bibliographystyle{IEEEtran}
{\small{
\bibliography{sc_gamp}
}}
\end{document}